\documentclass[aps,pra, twocolumn]{revtex4-1}
\usepackage{amsthm}

\usepackage[utf8]{inputenc}
\usepackage[sc,osf]{mathpazo}
\usepackage{amsmath}
\usepackage[T1]{fontenc}
\usepackage{latexsym}
\usepackage{amssymb}
\usepackage[colorlinks=true,citecolor=blue,urlcolor=blue]{hyperref}
\usepackage{color}
\usepackage{graphics,epstopdf}
\usepackage{soul}
\usepackage[demo]{graphicx}
\usepackage{capt-of}
\usepackage{lipsum}
\usepackage{adjustbox}
\usepackage[normalem]{ulem}
\usepackage[table,xcdraw]{xcolor}

\newcommand{\be}{\begin{equation}}
\newcommand{\ee}{\end{equation}}
\newcommand{\ba}{\begin{eqnarray}}
\newcommand{\ea}{\end{eqnarray}}

\newtheorem{theorem}{Theorem}

\begin{document}
	
%\title{Improving  Dense Coding Capacity and Teleportation Fidelity of Random Quantum States Via Pre-processing}
\title{Performance of Dense Coding and Teleportation for Random States --\\  Augmentation via Pre-processing}

\author{Rivu Gupta$^{1}$, Shashank Gupta$^{2}$, Shiladitya Mal$^1$, Aditi Sen (De)$^1$}

\affiliation{$^1$ Harish-Chandra Research Institute and HBNI, Chhatnag Road, Jhunsi, Allahabad - 211019, India}
\affiliation{$^2$ S. N. Bose National Centre for Basic Sciences, Block JD, Sector III, Salt Lake, Kolkata - 700 106, India}
\date{\today}
\begin{abstract}

In order to understand the resourcefulness of a natural quantum system in quantum communication tasks, we study the dense coding capacity (DCC) and teleportation fidelity (TF) of Haar uniformly generated random multipartite states of various ranks. We prove that when a rank-2 two-qubit state, a Werner state, and a pure state possess the same amount of entanglement, the DCC of a rank-2 state belongs to the envelope made by  pure and  Werner states. In a similar way, we  obtain an upper bound  via the generalized Greenberger-Horne-Zeilinger  state  for rank-2 three-qubit states when the dense coding with two senders and a single receiver  is performed and entanglement is measured in the senders:receiver bipartition. The normalized frequency distribution of DCC  for randomly generated  two-, three- and four-qubit density matrices with global as well as local decodings at the receiver's end  are reported. The estimation of mean DCC for two-qubit states is found to be in  good agreement  with the numerical simulations.  Universally, we observe that the performance of random states for dense coding as well as teleportation decreases with the increase of the rank of states which we have shown to be surmounted by the local pre-processing operations performed on the shared states before starting the protocols, irrespective of the rank of the states. The local pre-processing employed here is based on  positive operator valued measurements  along with classical communication and we show that unlike dense coding with two-qubit random states, senders' operations are always helpful to probabilistically enhance the capabilities of implementing dense coding as well as teleportation.  

% We found the exact preprocessing operators that optimises the DCC and TF of all randomly generated pure two qubit states. We find the exact condition for mixed states that determine whether the \textit{hidden} DCC gets activated after the action of local preprocessing. We also demonstrate the equivalence between the analytical and numerically determined mean DCC. 
%In case of teleportation fidelity, we find that most of the rank-3 (93\%) and rank-4 (98\%) states that give non-classical TF are bell-local states . We find that although on average the random two qubit mixed states are not useful for teleportation and dense coding quantified by mean fidelity and mean coding capacity but they become useful after the action of local preprocessing operation. We call this the \textit{hidden} feature of random states. 

\end{abstract}

\maketitle
\section{Introduction}
\label{sec:intro}

The basic quantum information processing tasks like dense  coding \cite{bennettwiesner, Aditi} and teleportation \cite{teleoriginal} demonstrate the usefulness of quantum entanglement \cite{HoroRMP} in the field of quantum information science. In particular, the idea of dense coding  (DC) is to employ prior quantum  correlation between the sender and the receiver for enhancing classical message-carrying capacity while  in teleportation, unknown state gets transferred to a remote location without physical transportation  with the help of a shared entangled state and two bits of classical communication. Performance of dense coding which is dubbed as the dense coding capacity (DCC) of the shared channel  is quantified by the number of messages in a unit of bits carried from the sender to the receiver \cite{Bose, Hiroshima, Bruss, Ziman03}. On the other hand, in teleportation, the relevant figure of merit is the teleportation fidelity (TF) which measures the closeness between the state obtained by the receiver and the target state to be teleported at the sender's end \cite{teleoriginal,  Pawel, Horo96}. Over the years,  spectacular experiments have been performed to realize both the protocols   by using photons, massive particles, nuclear magnetic resonance, etc. \cite{Pirandola, atomicDC, DCexpphoton, DCNMR, Satexp}. 

After their inception, these two protocols have been generalised in many ways. Going beyond bipartite scenario, dense coding has been extended to a scenario of multiple senders and multiple receivers, which enlarges the possibility of encoding-decoding strategies in various ways \cite{Aditi, DCCamader, Prabhu,Tamoghna,  DCCTamo1}. In the case of multiple senders, it was shown that invoking more general encoding than unitary, collective encoding is better than the individual encoding \cite{Horodecki_2012} while for multiple receivers situated at far-apart  locations, locally  accessible information \cite{Localaccess} plays a crucial role to obtain the DCC when receivers are allowed to perform local operations and classical communication (LOCC) for decoding \cite{Aditi, DCCTamo1}. Similarly, original teleportation protocol which is commonly known as standard teleportation scheme (STS) also has been generalised which include telecloning \cite{telecloning}, multi-port  teleportation \cite{portPRL, portPRA, portMichal}, teleportation with multiple sender-receiver pairs \cite{Aditicap, Saptarshicap}, counterfactual teleportation \cite{counter}, reusing teleportation channel \cite{weaktele}.

In a realistic situation, ideal conditions to achieve perfect DCC and TF are never met due to noises in the channel and imperfections in the apparatuses. To circumvent this, Bennett \emph{et al.} proposed a method of distillation \cite{dist1, dist2, dist3, dist4} which is a collective pre-processing scheme involving many copies of shared  noisy entangled  states and  LOCC for  obtaining pure maximally entangled state, suitable for perfect DC and teleportation. The problem with distillation is that it requires a large number of resources and successfully works only when singlet fraction is above some threshold value \cite{dist2}. In the context of teleportation,  this problem has been resolved by invoking filtering operation which acts at the single copy-level and can probabilistically  provide an output having high TF  \cite{dist4}. Surprisingly, it was also shown that for a certain class of states, a dissipative channel can activate teleportation power \cite{Badziag}. For two-qubit states, optimal teleportation protocol is known together with optimal filter \cite{Verstraete}. Very recently, it has been shown   that in higher dimension,  filtering can also be effective for revealing hidden teleportation power of shared Werner state \cite{Werner89} and a class of rank-deficient state used as channel  \cite{Liang}. In obtaining a quantum advantage in the dense coding protocol,  coherent information \cite{Schumacher} plays the role similar to singlet fraction in the case of teleportation and up to our  knowledge, improving the DCC of a channel via filtering has not been investigated as yet. 

On a different front,  randomly generated density matrices \cite{Kendon, Zyczkowskibook, Hayden, Enriquez, Pandit, Gross, Soorya, Ratul, Hastings} provide a vital tool for analyzing and studying the trends of typical states in state space. They not only arise naturally in the chaotic process \cite{Haake} but also can be generated in a systematic manner based on randomness in the outcome of quantum measurement \cite{Ratul}. Moreover, against the intuition of observing random behavior, it has been found that random states show some universal properties -- average quantum correlations among randomly generated states increase with the increase of the number of parties \cite{dist4, Badziag, Verstraete, Werner89}.  Random states were instrumental in disproving a long-standing conjecture in quantum information science regarding additivity of minimal output entropy \cite{Schumacher} and in showing constructive feedback in presence of a non-Markovian noisy environment \cite{Rivurecent}.

In the present work, we investigate the patterns of  capabilities obtained from two prominent quantum communication tasks  for Haar uniformly generated random shared channels. In particular, we estimate the distributions of the dense coding capacity of states having different ranks  in three specific scenarios: (1) a single sender and a single receiver, (2) two senders and a single receiver, and (3) two senders and two receivers. Note that in the first two cases, the decoding is done by global operations while in the third situation, the encoded states can only be decoded via LOCC. We prove  that the DCC of a rank-2 two-qubit state lies in  the envelope of the DCC of a pure state and  a Werner state when all of them possess the same amount of entanglement.  We numerically confirm  that such upper and lower bounds hold also for rank-3 and -4 two-qubit states. On the other hand, we show that  when three-qubit generalized Greenberger-Horne-Zeilinger (gGHZ) \cite{GHZ}  and a rank-2 state have the same amount of entanglement in the senders:receivers bipartition, the DCC of the gGHZ state is  higher than that of the rank-2 three-qubit state.  The mean of the frequency distribution for DCC is obtained numerically for random states which are shown to be in good agreement with analytical estimation. In all  scenarios of DC and teleportation protocols, we observe that the efficiencies  decrease with the increase of rank for the random states.  We apply  local pre-processing operations in the form of dichotomic positive operator valued measurements (POVM)s on the shared state before starting the protocol and report that the performance can be enhanced by such pre-processing mechanism for random states. Specifically, by employing three kinds of figures of merit, we establish that the local pre-processing at the sender's or the receiver's  or both the ends can help to probabilistically improve  the capacities, as well as the  teleportation fidelities, especially in higher ranked random states. One should note here that the pre-processing operations exploited here cannot be included in the encoding-decoding strategies (cf. \cite{Horodecki_2012}).

The paper is organized in the following way. In Sec. \ref{sec:Prelim}, we recapitulate the generation of random states of different ranks, the dense coding capacity, the teleportation fidelity, and the general dichotomic local POVM elements for pre-processing. In Sec.  \ref{sec:DCCbefore}, we provide our analytical results and numerical observations on dense coding capacity before pre-processing while the results obtained after local pre-processing is presented in Sec. \ref{sec:DCCafter}.    In Sec. \ref{TF},  observations and results on teleportation fidelity before and after pre-processing are reported. Finally, we conclude with a summary of results in Sec. \ref{Conclusion}.

%Most of the previous works in this direction consider specific class of states. We removed this constraint by considering random states that appear naturally in a chaotic system and give universal understanding of such quantum communication tasks. 
%
%In the present work, we study the two important quantum communication tasks (QCTs) using Haar uniformly generated random states. In particular, we estimate the distribution of dense coding capacity of different rank states in three specific scenarios: (1) one sender and one receiver, (2) Two senders and one receiver and (3) Two senders and two receivers. We find that the two qubit Werner states lower bound the capacity of rank-2 states in scenario-(1) and three qubit GHZ states upper bound the capacity of rank-2 states in scenario-(2).  In general, a random state does not have non-classical DCC or TF. So, we apply the local pre-processing operations to activate the hidden DCC or TF in such states.
%
%The paper is organized in the following way. In Sec.(\ref{Prelim}), we recapitulate the generation of random states of different ranks, the dense coding capacity and teleportation fidelity and the general dichotomic local POVM elements for pre-processing. In Sec.( \ref{DCC}), we provide our analytical results and numerical observations on dense coding capacity while  in Sec.(\ref{TF}) observations and results on teleportation fidelity are presented. Finally, we conclude with a summary of results in Sec.(\ref{Conclusion}).

\section{Definitions: Dense coding capacity and teleportation capability of multipartite random states}
\label{sec:Prelim}

In this section, we briefly describe dense coding  capacity involving  an arbitrary number of  senders and a single as well as two receivers and  define the teleportation fidelity for two-qubit states. Since we perform DC and teleportation for randomly generated states, let us first elucidate the procedure for such simulations \cite{Zyczkowskibook}. Haar uniform generation of pure states with arbitrary number of parties having complex coefficients, \(x_i = a_i + i b_i\), (\(a_i\) and \(b_i\)s are real numbers), where real numbers are taken from a Gaussian distribution with mean $0$ and standard deviation unity, denoted by \(G(0,1)\) is performed. Random mixed states of various ranks can be obtained from an appropriate multipartite pure state after taking partial traces of  suitable subsystem. For example,    two-qubit density matrices with rank-2, -3, and -4  can be simulated from random tripartite pure states  chosen in complex Hilbert spaces of $C^2\otimes C^2 \otimes C^2$, in \(C^2\otimes C^2 \otimes C^3\) and \(C^2\otimes C^2 \otimes C^4\) respectively \cite{Zyczkowskibook}.    

%Using these randomly generated states, we study two different quantum communication tasks, namely dense coding capacity to investigate how much classical information can be shared using random states and teleportation fidelity to study the efficiency in creating arbitrary qubit at a distant lab using classical communication and shared random states. Let us first recall dense coding capacity and then teleportation fidelity.

\subsection{Dense Coding Capacity}
\label{Prelim_DCC}

Consider a multipartite communication channel formed by multiple senders, \(S_1, S_2, \cdots, S_N\) and a single receiver, \(R\) in which classical information transmission via quantum states occurs. 
% where the Alices act as senders and  Bobs as receivers. 
As originally proposed by Bennett and Weisner \cite{bennettwiesner}, it can be shown that  if senders and a receiver apriori share an entangled state, \(\rho^{S_1 S_2\ldots S_N R}\), more bits of classical information can be encoded and sent to the receiver compared to a protocol with unentangled states.  The maximum classical information accessible by the receiving party is called the dense coding capacity  \cite{Bose, Hiroshima,  Ziman03, Bruss, DCCamader}. 
We consider two scenarios depending on the number of senders and receivers as mentioned earlier. 

\textit{1) \(N\)-senders and a single receiver (NS-1R).} Suppose \(N\)-senders and a single receiver share  an (N+1)-party quantum state, $\rho^{S_1 \cdots  S_{N}R}$. The dense coding capacity in this case reads as
\begin{align}
&& \mathbb{C}^{NS-1R} (\rho^{S_1 \cdots S_{N}R}) = \max[\log_2 d_{S_1} + \cdots + \log_2 d_{S_{N}}, \nonumber \\
&& \log_2 d_{S_1} + \cdots + \log_2 d_{S_{N}}  + S(\rho^{R}) - S(\rho^{S_1 \cdots S_{N}R})]
\label{DCC_formula1}
\end{align}
where, $d_{S_1}, \cdots ,d_{S_N}$ are the dimension of the senders' subsystems, $S_1, \cdots, S_N$, respectively.  $\rho_{R} $
%= \text{Tr}_{A_1 \cdots A_N} \rho^{A_1 \cdots A_N B_1}$, 
is the reduced state at the receiver's end   and \(S(\sigma) = - \mbox{tr} (\sigma \log_2 \sigma)\) is the von-Neumann entropy. The first term represents the amount of classical information that can be sent only by using classical protocol while a quantum state is suitable for dense coding if \( S(\rho^{R}) - S(\rho^{S_1 \cdots S_{N}R}) >0\). Notice that for two-qubits involving a single sender and a single receiver (1S-1R), it reduces to \(1 + S(\rho^R) - S(\rho^{SR})\). 

\textit{2) \(N\)-senders and \(2\)-receivers (NS-2R).} Let us now consider that there are two receivers, \(R_1\) and \(R_2\), in the dense coding protocol which again involve arbitrary number of senders, $S_1, S_2, \cdots, S_N$, sharing an (N+2)-party quantum state,  $\rho^{S_1 \cdots  S_{N}R_1R_2}$.  In this situation, although we do not know the exact DCC, the upper bound is known \cite{DCCamader}. 
%Consider \(N\) Alices, , and two Bobs, $R_1$ and $R_2$  share  
% is shared between  the \(N\) senders,  $S_1, \cdots, S_N$ and two receivers. 
Let 'k' senders, $S_1, \cdots, S_k$, send their parts of the shared state to the first receiver, $R_1$, while the remaining senders, $S_{k+1}, \cdots, S_{N}$, send their states to the second one, $R_2$. The upper bound on the dense coding capacity  is then represented by  
\begin{align}
&\mathbb{C}^{NS-2R}(\rho^{S_1 \cdots S_{N}R_1R_2}) \leq \max [\log_2 d_{S_1} + \cdots + \log_2 d_{S_{N}}, \nonumber \\ 
&\log_2(d_{S_1}) + \cdots + \log_2(d_{S_{N}}) + S(\rho^{R_1}) \nonumber \\ 
&+ S(\rho^{R_2})- \text{max}(S(\rho^{S_1 \cdots S_{k}R_1}), S(\rho^{S_{k+1} \cdots S_{N}R_2}))] \equiv U^{NS-2R}
\label{DCC_formula2}
\end{align}
where $\rho^{R_1} = \text{tr}_{S_1 \cdots S_NR_2} \rho^{S_1 \cdots S_N R_1R_2}$ and $\rho^{R_2} = \text{tr}_{S_1 \cdots S_NR_1} \rho^{S_1 \cdots S_N R_1R_2}$ are the reduced states of the first  and the second receiver respectively. Similarly,  \(\rho^{S_1 \cdots S_{k} R_1} = \text{tr}_{S_{k+1} \cdots S_N R_2} \rho^{S_1 \cdots S_N R_1 R_2} \)  and \(\rho^{S_{k+1} \cdots S_{N}R_2} = \text{tr}_{S_{1} \cdots S_{k} R_2} \rho^{S_1 \cdots S_N R_1 R_2} \). We will investigate the behavior of   the upper bound for the Haar uniformly generated four-qubit states where there are two senders and two receivers.

\subsection{Teleportation Fidelity}
\label{Prelim_TF}

In the teleportation protocol, the task is to send an unknown quantum state  to the receiver. If a shared state is maximally entangled, such a task can be accomplished by performing the entangled measurements at the sender's side and communicating the outcomes to the receiver. 
 %two bits of classical communication \cite{teleoriginal}.  
 Let us suppose that the sender, Alice and  the receiver, Bob share an arbitrary  bipartite state $\rho^{SR}$. The teleportation fidelity of $\rho^{SR}$ can be expressed as \cite{Horo96, Pawel}
\begin{equation}
	\mathbb{F} = \frac{\text{d}f + 1}{\text{d}+1}
\end{equation}
where $f = \underset{\{\phi\}}{\text{max}} \langle \phi| \rho^{SR} |\phi \rangle$ with $\{\phi\}$ being the set of all maximally entangled two-qudit states and \(d\) is the dimension of the input state to be teleported. Notice that Alice and Bob have the freedom to apply any trace-preserving local quantum operations and classical communication (LQCC) in order to maximize $f$ which is, in general, hard to perform even numerically. 

Given a two-qubit state, $\rho^{SR}$, we can calculate the optimal teleportation fidelity by using the Horodecki's prescription \cite{Horo95}, i.e.,
\begin{equation}
	\mathbb{F}_{max} \leq \frac{1}{2}(1+\frac{1}{3} \text{tr}\sqrt{C^{\dagger}C})
	\label{TF_Horo}
\end{equation}
where the elements of the matrix, $C = [C_{ij}]$, are given by $C_{ij} = \text{tr}[\rho^{SR} (\sigma_i \otimes \sigma_j)]$, where $\sigma$'s are the Pauli spin matrices. Furthermore, if the state $\rho^{SR}$ violates the Clauser-Horne-Shimony-Holt  inequality \cite{Bell, CHSH}, i.e., if it satisfies $M(\rho^{SR}) > 1$ \cite{Horo95}, where $M(\rho^{SR}) =(u_1 + u_2)$ with  $u_1$ and $u_2$ being the highest two eigenvalues of the matrix $C^{\dagger}C$,  the inequality (\ref{TF_Horo}) is replaced by an equality. 

\subsection{Preprocessing Operations}

%Random states do not have the optimal DCC or TF. 
We know that the initial DCC or TF of a state can probabilistically  be increased if some or all 
%In order to increase these features, 
of the parties  apply local pre-processing operations \cite{Perescoll, Gisin, Horodecki_2012, Verstraete}. If  the DCC (TF) is initially in the classical region and after pre-processing, it gives a quantum advantage, we say that the state exhibits \textit{hidden} DCC (TF). If the initial state already shows quantum advantage in  dense coding (teleportation), and after preprocessing, the advantage gets improved with some positive probability,  those states demonstrate \emph{enhancements} in DCC (TF). 
For the present study, we apply the most general dichotomic POVMs \cite{Busch, YuOh, Shila2016, Shila2017} as local pre-processing operations. \\
% which can occur due to the imperfections \\in the measurements performed both by the senders and the receivers on randomly generated states. \\
	$\bullet$ {\bf General dichotomic POVMs:} The general dichotomic POVMs can be represented as
\begin{equation}
	E_i^{\pm} = \lambda P_i^{\pm} + \frac{1 \pm \gamma_i - \lambda_i}{2}I
	\label{POVM_i}
\end{equation}
where $\lambda_i$ is the sharpness parameter, such that $0\leq \lambda_i \leq1$, $|\lambda_i| + |\gamma_i|\leq 1$ and $E_i^+ + E_i^- = \openone$, with \(\openone$ being the identity operator. $P_i^{+} = \cos \frac{\theta_i}{2} |0\rangle +  e^{ i \phi_i} \sin \frac{\theta_i}{2} |1\rangle$ and its orthogonal  projector is $P_i^-$. Here, i represents the party which applies the POVM. 
% obtained from the following states 
%\begin{eqnarray}
%|\psi_A \rangle_1 = cos(\frac{\theta_A}{2})|0\rangle +  e^{i\phi_A}sin(\frac{\theta_A}{2})|1\rangle \:and \nonumber \\
%|\psi_A\rangle_2 = sin(\frac{\theta_A}{2})|0\rangle -  e^{i\phi_A}cos(\frac{\theta_A}{2})|1\rangle
%\end{eqnarray}
%with $0<\theta_A<\pi$ and $0<\phi_A<2\pi$. Similarly, the general dichotomic POVMs associated with Bob are defined. \\
To find the optimal POVM, we have to maximize over the set of parameters, \(\{\theta_i, \phi_i, \lambda_i\}\). If the shared state is two-qubits and both the parties perform local preprocessing before starting the protocol,  we have to carry out maximization over  six parameters  to evaluate maximal  DCC (TF). 
%Depending on the quantum communication task, pre-processing is done. For example, in dense coding, some sender(s) or receiver(s) or all of them can perform local pre-processing operations to activate the hidden DCC/TF. As an example, 
In a multipartite shared state, $\rho^{S_1 \cdots S_N R}$, considering  preprocessing operations performed by first $k$ senders,  the output state  after the action of local POVMs is  given by
\begin{widetext}
\begin{equation}
 	\rho^{S_1 \cdots S_NR}_P = \frac{(\sqrt{E_{S_1}^{\pm}} \otimes \cdots \otimes \sqrt{E_{S_k}^{\pm}}\otimes \openone_{S_{k+1}} \otimes \cdots \otimes \openone_{R}) \rho^{S_1 \cdots S_NR} (\sqrt{E_{S_1}^{\pm\dagger}} \otimes \cdots \otimes \sqrt{E_{S_k}^{\pm\dagger}} \otimes \openone_{S_{k+1}} \otimes \cdots \otimes \openone_{R})}{\text{tr}[(\sqrt{E_{S_1}^{\pm}} \otimes \cdots \otimes \sqrt{E_{S_k}^{\pm}}\otimes \openone_{S_{k+1}} \otimes \cdots \otimes \openone_{R}) \rho^{S_1 \cdots S_NR} (\sqrt{E_{S_1}^{\pm\dagger}} \otimes \cdots \otimes \sqrt{E_{S_k}^{\pm\dagger}} \otimes \openone_{S_{k+1}} \otimes \cdots \otimes \openone_{R})]}.
	\label{POVM_state1}
\end{equation}
\end{widetext}
Notice that the DCC (TF) of the resulting state is investigated after maximizing over $3k$ parameters involved in \(k\) local POVMs. 

\section{Dense Coding Capacity of Random States without preprocessing}
\label{sec:DCCbefore}

Let us first present the behavior of dense coding capacity of Haar uniformly generated multipartite states. In particular, we analyze the frequency distributions in three scenarios. 

\emph{Case 1.} A single sender - a single receiver (1S-1R) pair shares  two-qubit random states with different ranks. \\
\emph{Case 2.} Two senders and a single receiver (2S-1R) have three-qubit Haar uniformly generated states having rank-1, -2, -3, -4, -5 and -6. \\
\emph{Case 3.} Haar uniformly simulated four-qubit pure as well as states having rank -2, - 3 and -4 are initially distributed among two senders and two receivers (2S-2R) situated in distant locations. 
%For the case two senders-one receiver (2S-1R) sharing a three qubit random state, we analyse the frequency distribution of DCC of random states upto rank-4 only where as in case of two senders-two receivers (2S-2R) sharing a four qubit random state, we study DCC of rank-1 and 2 only. This is because interesting quantum effects present or hidden are extractable with meaningful probability for such states only.

Before proceeding further, note that  quantum advantages are not obtained when the DCC, \(\mathbb{C}\), is  unity for 1S-1R, two for both 2S-1R and 2S-2R situations provided the dimension of each party is restricted to be two.  In cases of 1S-1R and 2S-1R, if the shared state is pure,  the second term in Eq. (\ref{DCC_formula1}) vanishes and hence DCC reduces to the von Neumann entropy of the receiver's part. Since the entanglement of a pure bipartite state can  uniquely be quantified by the von Neumann entropy of the local density matrix \cite{dist1}, the quantum advantage  can always be achieved for all entangled pure states. For mixed two-qubit states, we will show below by proving a theorem that such a connection between shared entanglement and DCC cannot be established. 

In a 2S-1R case, a relation between genuine multipartite entanglement content of the shared pure state and DCC does not hold \cite{DCCTamo1} and to our knowledge, no such results are known for mixed three-qubit states which will also be established for rank-2 three-qubit states.  In this work, we also concentrate on mixed three-qubit states with rank upto six. 

%Before going into detailed behavior of random states, let us first mention some of the interesting results in this direction. It was realized from different studies that higher the entanglement, higher is the coding capacity of the state in one sender and one receiver case. In multiparty scenario, involving many senders and receiver, it is not true because there exist states with lesser value of genuine entanglement but more coding capacity in comparison to the states having higher value of genuine entanglement. It was realized that pure states in general have more coding capacity than mixed states. 

Let us first investigate the behavior of dense coding capability of random states. Entire calculations and analysis are  based on \(5\times 10^4\) Haar uniformly generated states for each case. In all these scenarios, the normalized frequency distribution of DCC, given by \(\mathcal{F}_{DC} = \frac{N_{DC} (C(\rho))}{N_S}\), with \(N_{DC} (C(\rho))\) being the number of states having DCC \( C(\rho)\) and  \(N_S\)  being the total number of simulated states, is calculated except in the situation of  two senders and two receivers case where the normalized distribution of the upper bound is analysed, as depicted in Fig. \ref{fig:distriDCC}. The observations  in the figure are listed below:

\begin{enumerate}
\item Obtaining a quantum advantage in the DC protocol  decreases with the increase of the rank of the states. It can be argued that such behavior is seen because the average entanglement content of the Haar uniformly generated states decreases with the rank of the states. However, such a simple explanation may not hold as we will show below. 
% As mentioned before,   DCC and entanglement content of the shared state is not related for mixed two-qubit states as well as for multipartite states. 

\item Percentages of states showing DCC more than the classical bound are  \(50.09\%\), \(4.80\%\) and \(0.30\%\) for two-qubit states with rank-2, -3, and 4 respectively. For the 2S-1R DC scheme, it turns out to be \(50.31\%\), \(0.08\%\)  for rank-2 and  rank-3 states while no states are found to give a quantum advantage from rank$\geq$4 random states. All pure states are good for classical information transmission. 

\item The upper bound in the 2S-2R case showing quantum superiority is  seen for  \(97.36\%\)  of rank-2 four-qubit states and  for all pure random states. For higher ranks, unlike 2S-1R DC protocol, the above percentage decreases but remains significant, being \(95.77\%\) and \(95.34\%\) for ranks -3 and -4 respectively.
% can give more than the classical bound.  

\item The pattern of  \(\mathcal{ F }_{DC}\) also changes with rank as well as with the increase in the number of  senders and receivers. Specifically, we observe that the fraction  of states, showing nonclassical capacity decreases with the rank of the random states as shown in Fig. \ref{fig:distriDCC} (lower panel  (right)), irrespective of the DC schemes. It can also be captured by computing the mean and standard deviation (SD) of the distribution which we will discuss in the succeeding sections. We will also study how the distribution changes with the introduction of preprocessing in terms of POVM by different figures of merits. 

\end{enumerate}
%We observed the similar behavior as shown in fig. (\ref{distriDCC}). We found that the maximum DCC decreases with increasing rank for every case. We found that all pure states have non-classical value of DCC for all three cases. This is in accordance with the previous studies.  In 1S-1R case, around {\color{blue}40$\%$} states have DCC in the middle region (1.3-1.7) whereas in 2S-1R case, {\color{blue}42$\%$} states have DCC in the extreme non-classical region (2.7 or more). DCC distribution in 2S-2R case is similar to 1S-1R case. We find that most high rank states (rank-3 and 4) have DCC in classical region. Later on we demonstrate the action of local pre-processing operations in revealing the hidden DCC of such states. 
\begin{figure}[!ht]
	\resizebox{8cm}{7cm}{\includegraphics{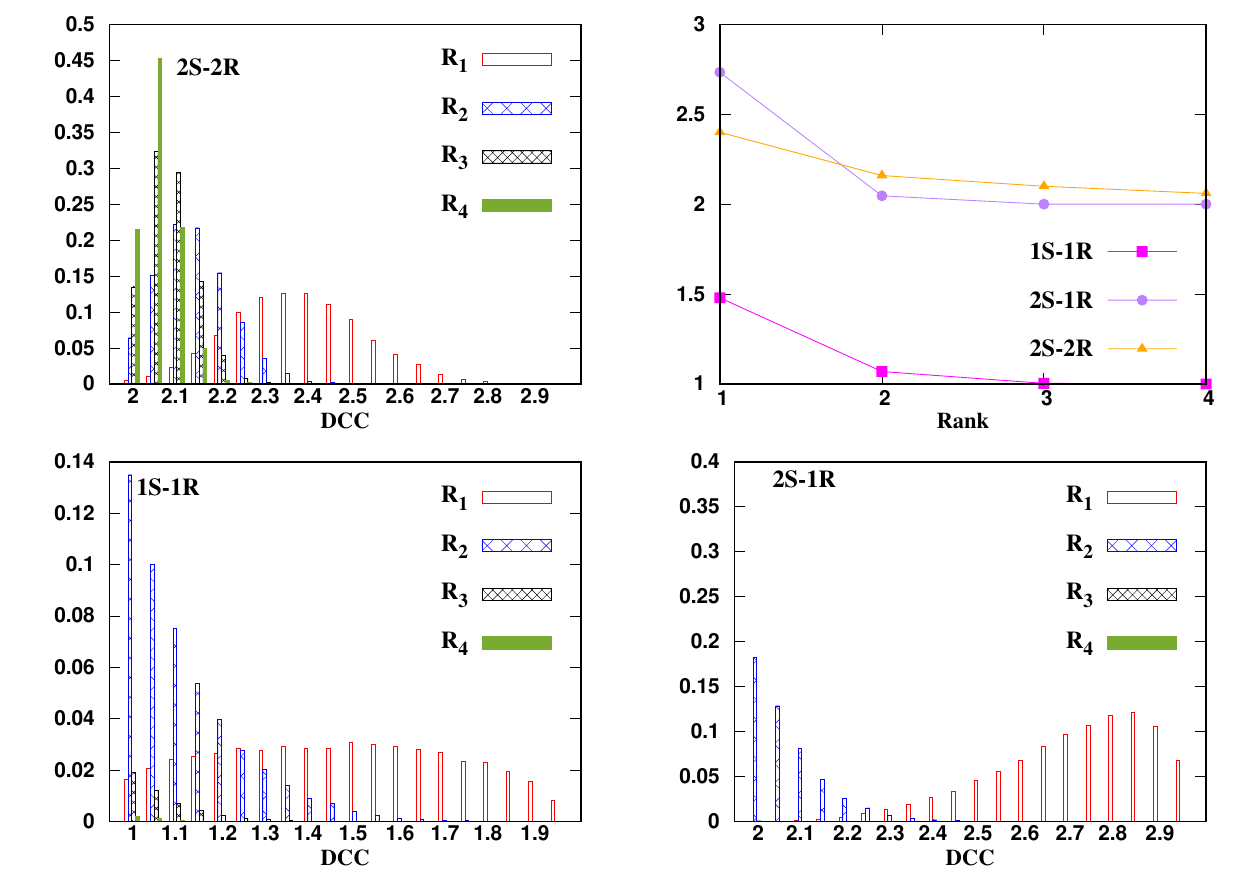}}
	\caption{\footnotesize (Color online)  (Upper panel) and (Lower panel (left)). The normalized frequency distribution, \(\mathcal{F}_{DC}\),  of Haar uniformly generated states (vertical axis) against  DCC (horizontal axis). Upper Panel.  (Left) A single sender-single receiver (1S-1R) and  (Right) single sender-two receivers (1S-2R) scenarios. Lower panel (left) two senders - two receivers (2S-2R).  \(R_1, \ldots R_4\) denotes the random states of rank-1 to rank-4. Lower Panel (right). Fraction of states  having quantum advantage in dense coding  vs. the rank of random states for three DC protocols. 	Notice that the large fraction of high rank mixed states have DCC in classical region and the general tendency to have quantum advantage decreases with the increase in rank. The rate of decrease of the upper bound for DCC with rank  in 2S-2R case is significantly slower than the rate of DCC for 1S-1R and 2S-1R.  All the axes are dimensionless. }
\label{fig:distriDCC}
\end{figure}

To establish the fact that for mixed bipartite states, DC and entanglement content is not related, we will now show that the  DCC of random states has a universal lower bound. In particular, we find that the DCC capacity of the Werner state, given by 
\begin{equation}
\rho_W = p|\phi^+\rangle \langle \phi^+ | + \frac{(1-p)}{4} I_4
\label{Werner}
\end{equation}
where $|\phi^+\rangle = \frac{1}{\sqrt{2}} (|00\rangle + |11\rangle)$ with $0\leq p \leq 1$ and $I_4$ being the identity matrix in $C^2 \otimes C^2$,  gives a lower bound for all randomly generated two-qubit states of rank-1 to rank-4  ( see Fig. \ref{R2_cap_werner} ). Moreover, we observe that the DCC of Haar uniformly generated states with rank-2, -3 and -4 lies between the envelopes obtained for pure states, and the Werner states.
% provided states possess a moderate amount of entanglement.
 %We interestingly observe that when the entanglement content of a pure state and rank-2 states are low, there exists some rank-2 states which has higher DCC than that of the pure state having the same amount of entanglement. 
 Let us now prove the lower and upper bounds  for rank-2 states. 

\begin{figure}[!ht]
	\resizebox{9cm}{8cm}{\includegraphics{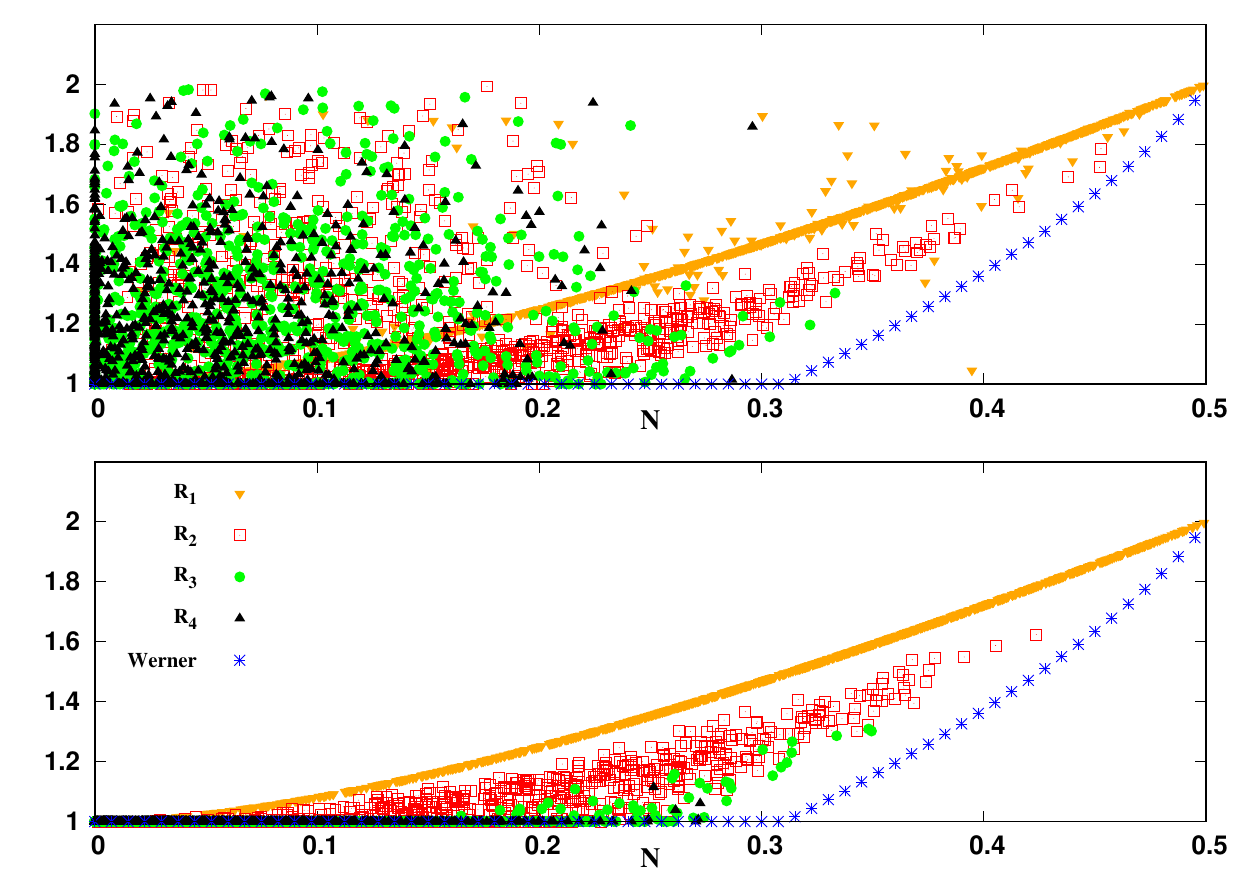}}
	\caption{(Color online) Lower panel.  Dense coding capacity of randomly generated two-qubit states  (vertical axis) against entanglement (horizontal axis) which is quantified by negativity.  Blue line represents the Werner state, \(\rho_W\) while the orange line represents the two-qubit pure state. Upper panel. The maximal cost of average  DCC, defined later in Eq. (\ref{averagecost2_DCC}) of random two-qubit states of different rank after two-sided POVMs is plotted with respect to negativity of the given initial state. We notice that the lower bound still holds after local POVMs applied by both the parties. The vertical axis is in bits while the horizontal axis is in ebits. }
\label{R2_cap_werner}
\end{figure}

\begin{theorem} 
The dense coding capacity of the arbitrary mixed two-qubit state of rank-2 in the 1S-1R case is upper bounded by the capacity of a pure state and lower bounded by a two-qubit Werner state when all of them possess the same amount of entanglement. 
\end{theorem}

\begin{proof}
%\textcolor{red}{how have u shown for all ranges parameters? Cud u pls explain?}
%\textit{Proof:}  
Any two-qubit mixed state of rank-2 can be expressed as \cite{WHSV}
\begin{equation}
\rho_2^2 = p_1 |\psi_1\rangle \langle\psi_1| + (1-p_1)|\psi_2\rangle \langle\psi_2|,
\label{GR2}
\end{equation}
where  $\: 0 <  p_1 < 1$,   $|\psi_1\rangle = |0\eta_1\rangle + |1\eta_2\rangle \:$, $\: |\psi_2\rangle = |0\eta_1^\perp\rangle + |1\eta_2^\perp\rangle$, $|\eta_1\rangle = \cos\frac{\theta_1}{2}|0\rangle + \sin\frac{\theta_1}{2}|1\rangle \:$ and $\: |\eta_2\rangle = \cos\frac{\theta_2}{2}|0\rangle + \sin\frac{\theta_2}{2}|1\rangle$ with $|\eta_1^\perp\rangle$ and $|\eta_2^\perp\rangle$ being orthogonal states to $|\eta_1\rangle$ and $|\eta_2\rangle$ respectively,
%All the coefficients are taken to be real for analytical simplicity and each of the states are normalized 
%We take 
and  $0\leq \theta_i \leq \pi,\, i=1,2$.
% lies between $0 \; \text{and} \; \pi$.
The entanglement here is quantified by the negativity \cite{Vidal,negZ1,negZ2}  which is defined as  the sum of the modulus of negative eigenvalues of the partially transposed state. In this case, negativity  of $\rho_2^2$ in Eq. (\ref{GR2}) reads as
\begin{eqnarray}
\mathbb{N}_2^1 &=& |\frac{1}{4} \left[ \sqrt{x} - 2(1 - p_1)\right]| \label{e2}\\
\mathbb{N}_2^2 &=& |\frac{1}{4} \left[ \sqrt{x} - 2p_1\right]|, \label{e3}
\end{eqnarray}
where $x = 2 + 4p_1(p_1 + 1) + 2(2p_1 - 1)\cos(\theta_1 - \theta_2)$. Note that for a fixed \(p_1, \theta_i\, (i=1,2)\), 
% only one of $\mathbb{N}_2^1 \; \text{or} \; \mathbb{N}_2^2$ represents negativity, i.e., $ 
$\mathbb{N} (\rho_2^2) = \text{max}\{0,\mathbb{N}_2^1,\mathbb{N}_2^2\}$. 

Let us first show the upper bound. The similar line of
proof leads to the lower bound. An arbitrary two-qubit pure state written in a Schmidt decomposition reads as
\begin{equation}
|\psi \rangle = \cos\frac{\theta}{2} |0_S 0_R \rangle + \sin\frac{\theta}{2} |1_S 1_R \rangle
\label{r1}
\end{equation}
where $|0_{S (R)}\rangle \text{ and } |1_{S (R)}\rangle$ are the eigenvectors of the reduced density matrices corresponding to the sender (receiver) and the eigenvalues of the local density matrix are $\cos^2\frac{\theta}{2} \; \text{and} \; \sin^2\frac{\theta}{2}$. The negativity of the pure state is the square root of the determinant of its reduced density matrix, i.e. $\sin\theta/2$. 
%We obtain  which quantifies the negativity of $\rho_2^2$
Equating entanglements of  rank-2 and pure state, we obtain
\begin{equation}
\theta = \sin^{-1}(2\mathbb{N}). 
\label{relation}
\end{equation}
%Substituing this in Eq. (\ref{r1}),  we form the rank 1 state having the same entanglement as the rank 2 one.\\
On the other hand, the DCC of $\rho_2^2$ can be written as 
\begin{eqnarray}
\mathbb{C} (\rho_2^2) &=& 1 + H(\{\frac{1}{2}(1 - f_1(p_1)), \frac{1}{2}(1 + f_1(p_1))\}) \nonumber \\
&-& H(\{p_1,1 - p_1\}),
\end{eqnarray}
where $H(\{p_i\}) = -\sum_i p_i \log_2(p_i)$ is the Shannon entropy of the probability distribution $\{p_i\}$, and $f_1(p_1) = \frac{1}{2}(1 - 2p_1)\cos(\frac{\theta_1 - \theta_2}{2})$ while $\mathbb{C} (|\psi \rangle) = 1 + H(\{\cos^2(\theta/2), \sin^2(\theta/2)\})$. Due to Eq. (\ref{relation}), \(\mathbb{C} (|\psi \rangle)\) turns out to be  a function of $p_1, \; \theta_1, \; \text{and} \; \theta_2$ which can help to prove the statement of the theorem, i.e., by showing inequality given by
\begin{eqnarray}
&& H(\{\frac{1}{2}(1 - f_1(p_1)), \frac{1}{2}(1 + f_1(p_1))\}) - H(\{p_1,1 - p_1\}) \nonumber \\ 
&& - H(\{\cos^2\frac{\theta}{2}, \sin^2\frac{\theta}{2}\}) < 0 
\label{ineq1}
\end{eqnarray}
We substitute the value of $\theta$ in terms of $p_1, \; \theta_1, \; \text{and} \; \theta_2$ using Eq. (\ref{relation}),  and   numerically find that the inequality in  (\ref{ineq1}) holds true for all values of the above parameters.

In a similar fashion, we find that the negativity of the Werner state, $\rho_W$, is $\frac{(1 - 3p)}{4}$.  If entanglements of $\rho_2^2$ and $\rho_W$ are equal, we get
% leads to the following relation between $p_1, \theta_1, \theta_2 \; \text{and} \; p$
\begin{equation}
p = \frac{1 - 4 \mathbb{N}}{3}. 
\label{relation2}
\end{equation}
The DCC of $\rho_W$ reads $1 + 1 + H(\{\frac{1 + 3p}{4}, \frac{1 - p}{4}, \frac{1 - p}{4}, \frac{1 - p}{4}\})$, since the local entropy of the reduced system of the Werner state is unity. By using Eq. (\ref{relation2}), we again numerically establish that DCC of any rank-2 state is
always higher than that of the Werner state when both of them possess the same amount of entanglement for all values of \(p_1, \theta_1\) and \( \theta_2\).

Notice that although the proof is presented for real parameters, we observe that if \(|\eta_i\rangle,\, i=1, 2\) also have complex coefficients, the proof holds.
\end{proof}

\textit{Remark 1.} Numerically, we find that both the bounds remain true for all two-qubit states with rank-3 and -4. 

\textit{Remark 2.}
Our numerical observations show that even after pre-processing, our theorem holds (see the upper panel  in Fig. \ref{R2_cap_werner}). It implies that when  the receiver or both sender-receiver pair apply the local POVMs to activate the  dense coding capability of  shared states, the DCC of a random two-qubit state is still  lower bounded by that of the Werner state having the same value of initial entanglement. 
%Note that POVM is applied in both cases, when the shared state is a rank-2 random state and also when it is a Werner state. 
However, the upper bound does not hold any more under local POVMs.

%\textit{Remark-2:} We observe that rank-3 and rank-4 states can have lower value of DCC than a Werner state at the same entanglement value. 

\begin{figure}[!ht]
	\resizebox{7cm}{6cm}{\includegraphics{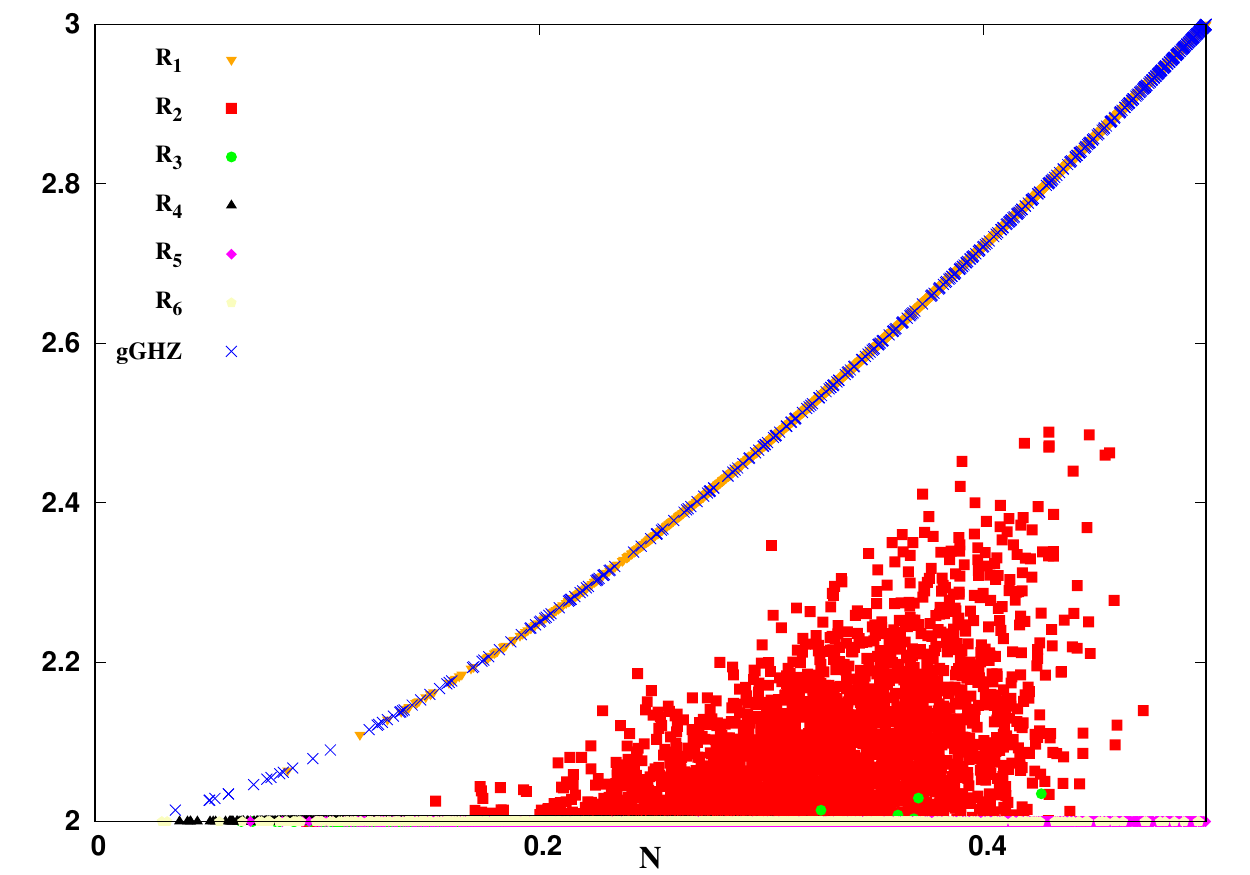}}
	\caption{(Color online) Dense coding capacity of Haar uniformly generated three-qubit states  (vertical axis) vs. negativity (horizontal axis) in the bipartition of senders and the receiver. Blue line represents the generalized GHZ state, \(|\phi_{gGHZ}\rangle\).   Subscripts, \(i\) (\( i=1,\ldots 6\)) of \(R_i\) denote the rank of the three-qubit states.  The vertical axis is in bits while the horizontal axis is in ebits.   }
\label{fig_GHZ3}
\end{figure}

\begin{theorem}
When negativities in the bipartition of senders and receivers of a three-qubit rank-2 state and the generalized GHZ state are equal,  the dense coding capacity of the latter is always higher than that of the former. 
\end{theorem} 

\begin{proof}
Any three-qubit rank-2 state, shared between \(S_1 S_2 R\) can be written as \cite{WHSV}
\begin{equation}
\rho_3^2 = p_2 |\psi_3\rangle \langle\psi_3| + (1-p_2)|\psi_4\rangle \langle\psi_4|,
\label{TR2}
\end{equation}
 where $0 <  p_2 < 1$, and  $|\psi_3\rangle = |0\eta_3\rangle + |1\eta_4\rangle \:$, $\: |\psi_4\rangle = |0\eta_3^\perp\rangle + |1\eta_4^\perp\rangle$, $|\eta_3\rangle = |0 \eta_3'\rangle + |1\eta_4'\rangle \:$ and $\: |\eta_4\rangle = |0\eta_3''\rangle + |1\eta_4''\rangle$, with  $|\eta_3'\rangle = \cos \frac{\theta_3}{2}|0\rangle + \sin\frac{\theta_3}{2} |1\rangle,  |\eta_4'\rangle = \cos\frac{\theta_4}{2}|0\rangle + \sin\frac{\theta_4}{2} |1\rangle, |\eta_3''\rangle = \cos\frac{\theta_3'}{2}|0\rangle + \sin\frac{\theta_3'}{2} |1\rangle, |\eta_4''\rangle = \cos\frac{\theta_4'}{2}|0\rangle + \sin\frac{\theta_4'}{2} |1\rangle$ with  $|\eta_3^\perp\rangle$ and $|\eta_4^\perp\rangle$ being orthogonal states to $|\eta_3\rangle$ and $|\eta_4\rangle$ respectively.
 %and $|0\eta_i\rangle$ is a short hand notation for $|0\rangle \otimes |\eta_i\rangle$. 
 %All the coefficients are taken to be real for analytical simplicity and each of the states are normalised.
 Here $0\leq \theta_i$, $\theta_i' \leq \pi, \, i=3,4$. 
The entanglement in terms of negativity \cite{Vidal}  of  rank-2 state in  the \(S_1S_2:R\) bipartition is
\begin{eqnarray}
  	(\mathbb{N}_3^2)^1 &=& \frac{1}{8}\left[\sqrt{12-24p_2 +16 p_2^2-(1-2p_2)y} - 4p_2\right], \nonumber\\
	&\text{if }& p_2 < 0.5 \label{e31} \\
 	(\mathbb{N}_3^2)^2 &=& \frac{1}{8} \left[ \sqrt{4-8p_2+16p_2^2 + (1-2p_2)y}+4p_2 -4 \right], \nonumber \\
	&  \text{if }& p_2 > 0.5 \label{e32}
 \end{eqnarray}
where $y = \cos(\theta_3 - \theta_3') + \cos(\theta_3 - \theta_4') + \cos(\theta_3' - \theta_4') + \cos(\theta_3 - \theta_4) + \cos(\theta_3' - \theta_4) + \cos(\theta_4' - \theta_4)$ and hence  $\mathbb{N}(\rho_3^2) = \max \{0, (\mathbb{N}_3^2)^1, (\mathbb{N}_3^2)^2 \}$ quantifies the negativity of $\rho_3^2$.
For the three-qubit generalized GHZ (gGHZ) state,
\begin{equation}
|\phi_{gGHZ} \rangle = \cos(\theta_g/2) |0_{S_1}0_{S_2}0_R\rangle + e^{i \phi_g} \sin(\theta_g/2)|1_{S_1}1_{S_2}1_{R}\rangle
\label{gGHZ}
\end{equation}
where $0 < \theta_g < \frac{\pi}{2}$ and \(0 \leq \phi_g \leq \pi\), 
%and $|0_{S_i/R}\rangle$, $|1_{S_i/R}\rangle$ are the eigenvectors of the reduced density matrices corresponding to the sender or receiver. 
we have \( \mathbb{N} (|\phi_{gGHZ} \rangle ) = \frac{\sqrt{1- \cos 2 \theta}}{2 \sqrt{2}}\). When the entanglements of both the three-qubit rank-2 and the gGHZ states coincide,  we find
\begin{equation}
	\theta_g = \frac{\cos^{-1} (1-8\mathbb{N}_3^2)}{2}
	\label{p_2relation}
\end{equation}
On the other hand, 
\begin{eqnarray}
	\mathbb{C}(\rho_3^2) &=& 2 + H(\{\frac{1}{2}(1 - f_2(p_2)), \frac{1}{2}(1 + f_2(p_2)) \} ) \nonumber \\
	&-& H(\{ p_2, 1 - p_2\}) 
\end{eqnarray}
where $f_2(p_2) = \frac{1}{2\sqrt{2}}\sqrt{(1-2p_2)^2 (2+y)}$ while  $\mathbb{C}(|\phi_{gGHZ}\rangle) = 2 + H(\{ \cos^2(\theta_g/2), \sin^2(\theta_g/2) \})$. 
%Using the relation from Eq. (\ref{p_2relation}), it boils down to a function of $p_2, \theta_3, \theta_4, \theta_3'$ and $\theta_4'$. 
Mathematically, the statement of the theorem requires the following inequality to hold
\begin{eqnarray}
       &&	H(\{\frac{1}{2}(1 - f_2(p_2)), \frac{1}{2}(1 + f_2(p_2)) \} ) - H(\{ p_2, 1 - p_2\}) \nonumber \\
	&&   - H(\{ \cos^2(\theta_g/2), \sin^2(\theta_g/2) \}) < 0.
	\label{gGHZ_in}
\end{eqnarray}
We substitute the value of $\theta_g$ in terms of $p_2, \theta_3, \theta_4, \theta_3'$ and $\theta_4'$ using the relation, given in Eq. (\ref{p_2relation}) and we numerically find that for all values of the above parameters, the inequality (\ref{gGHZ_in}) holds true.

\end{proof}  

%\begin{figure}[!ht]
%	\resizebox{9cm}{8cm}{\includegraphics{r2_cap_werner.pdf}}
%	\caption{\footnotesize (Coloronline) [Rank 2] Dense coding capacity of randomly generated rank 2 states is plotted against entanglement (in purple) and the Werner state is shown as a reference (in green)}
%\label{R2_cap_werner}
%\end{figure}

\textit{Remark 1. } Like in two-qubit states,  we also observe that the DCC of other mixed states of  rank$\geq$3 are also upper bounded by the DCC of the gGHZ state as shown in Fig. \ref{fig_GHZ3}. 

\textit{Remark 2.} Our numerical results show that after pre-processing, some of the rank-2 states have  higher  DCC than that of the gGHZ state when both of them possess the same amount of initial entanglement. Hence, our theorem does not hold when the senders and the receiver apply local POVMs.

\subsection{Analytical expression for mean DCC}
\label{sec:mean}

Let us now derive  the analytical expression of the mean dense coding capacity of Haar-uniformly generated two-qubit states of different ranks in the 1S-1R scenario. These analytical expressions match significantly well with our numerical results as obtained from numerical data in Fig.  \ref{fig:distriDCC}.
    %Let us now analytically derive the expression for mean DCC of  
    %, in terms of the dimensions of the subsystems.
    	From Eq. (\ref{DCC_formula1}), the mean DCC for random two-qubit states can be rewritten as
	\begin{eqnarray}
	\langle \mathbb{C}^{1S-1R} (\rho^{SR}) \rangle &=& 1 + \langle S(\rho^{R}) - S(\rho^{SR})\rangle. 
	%\langle \mathbb{C}^{A_1B_1}\rangle &=& 1 + \langle S(\rho^{B_1})\rangle - \langle S(\rho^{A_1B_1})\rangle
	\label{meanDCC}
	\end{eqnarray}
	%In the article \cite{Kendon}, Zyczkowski et. al, calculated 
	The mean entropy of a subsystem of dimension \(M\), which is obtained through partial tracing from a pure state of dimension \(MK\),  
	can be expressed as \cite{Kendon}
		\begin{equation}
	\langle S_M \rangle \approx \log_2 M - \frac{M}{2K}.
	\label{eq:wholeen}
	\end{equation}
	%where \(K\) is the dimension of the remaining subsystem. \\
	For arbitrary two-qubit states, $\rho^{SR}$, $M = 4$. Depending upon the rank of the system, \(K\)  can take value \(1 \ldots 4\) for states with  rank-\(1, \ldots 4\) respectively. In order to calculate $\langle S(\rho^{R})\rangle$ for the reduced state, we use the principle of purification of mixed states \cite{Hall} according to which a mixed state is obtained from a higher dimensional pure state after tracing out appropriate subsystems of a pure state.  
	If we assume that a  \(N+2\)- dimensional pure state leads to a single qubit state, the average entropy was found to be \cite{Karol} 
	%For a single qubit state composite with an \(N\)- dimensional state (the resulting pure state being N+2 dimensional), the average entropy is
%\begin{widetext}
	\begin{eqnarray}
	\langle S(\rho) \rangle = &&\frac{\log_2 e}{4^{N-1}} \frac{(2N - 1)!}{(N-2)! (N-1)!} \nonumber \\
	&&\sum_{s=0}^{N-2} \binom{N-2}{s} \frac{(-1)^s}{(s + 2)(2s + 3)} \sum_{t = 0}^{s+1} \frac{1}{2t + 1}
	\label{meanER}
	\end{eqnarray}
	Let us take $N = \frac{\mbox{dimension \: of \: initial \: pure \:  state}}{2}$. The results for different ranks are enumerated below. 
	
	\begin{table}[]
	\resizebox{0.5\textwidth}{!}{\begin{minipage}{0.7\textwidth}
			\caption{  Comparison between analytical  ( in Eq. (\ref{meanER})) and numerical values of $\langle \mathbb{C}^{SR} \rangle $.} 
			\label{tabmean}
			\centering
			\begin{tabular}{|r|r|r|r|r|}
				\hline
				\multicolumn{1}{|c|}{Rank} & \multicolumn{3}{c|}{Analytical}                                                            & \multicolumn{1}{c|}{Numerical} \\ \hline
				\multicolumn{1}{|c|}{}     & \multicolumn{1}{c|}{N = 2} & \multicolumn{1}{c|}{N = dim/2} & \multicolumn{1}{c|}{N = 100} & \multicolumn{1}{c|}{}          \\ \hline
				2                          & 0.481                      & 0.735                          & 1                            & 1                              \\ \hline
				3                          & 0.15                       & 0.489                          & 0.667                        & 0.711                          \\ \hline
				4                          & 0                          & 0.366                          & 0.5                          & 0.536                          \\ \hline
			\end{tabular}
	\end{minipage}}
\end{table}
%\end{widetext}
	%So, we can put different values of N for different ranked states. For rank 1 states, using N=2, this gives the exact result as obtained earlier. For ranks 2,3,4, with N = 2, we recover $\langle S(\rho_B) \rangle = 0.480898$. So, this formula is consistent with the above measure and in fact, it was used to derive $P_{HS}(r) = 24r^2$. 

\begin{enumerate}
	\item \textit{Pure states.} In this case, $\langle S(\rho^{SR})\rangle = 0$ and $\langle S(\rho^{R})\rangle = 0.5$ which gives $\langle \mathbb{C}_1^{SR} \rangle = 1.5$ The subscript in \(\mathbb{C}_1^{SR}\) denotes the rank of the state. 

	\item \textit{Rank-2 states.} $\langle S(\rho^{SR})\rangle = 1$ and  $\langle S(\rho^{R}) \rangle$ is calculated using $N = 4$ in the formula (\ref{meanER}) giving the value 0.735. Hence $\langle \mathbb{C}_2^{SR} \rangle = 	0.735$.
	
	\item \textit{Rank-3 states. } By using Eqs. (\ref{eq:wholeen}) and (\ref{meanER}) with $N=6$, we get $\langle S(\rho^{SR})\rangle = \frac{4}{3}$ and  $\langle S(\rho^{R}) \rangle = 0.822$ which leads to the mean DCC  as $\langle \mathbb{C}_3^{SR} \rangle = 0.489$.
	
	\item \textit{Rank-4 states.} In this case, $\langle \mathbb{C}_4^{SR} \rangle = 0.366$ since $\langle S(\rho^{SR})\rangle = \frac{3}{2}$ and  $\langle S(\rho^{R}) \rangle = 0.866$ which is obtained by using $N=8$. 
\end{enumerate}

\emph{Remark 1.} The average of DCC obtained in the case of two-qubit states having rank-2, -3, -4 is below unity which implies that most of the states do not give quantum advantage in the dense coding protocol and hence average DCC  decreases with the rank  (cf. \cite{Rivurecent}). \\
 
\emph{Remark 2.} Notice that the mean obtained by analysing the  frequency distributions of DCC in Fig. \ref{fig:distriDCC}  is much higher than the one reported above as also shown in Table \ref{tabmean}.   We find that if we increase the dimension of the composite system, \( N\), the analytical results match pretty well with the numerics. In fact with N = 100, the analytical and numerical results are in good agreement (see Table \ref{tabmean}). 

%\textcolor{red}{However, there is a tradeoff. Increasing N, decreases the number of distinct reduced density matrices in the random ensemble. Since the mean DCC of the mixed state is in classical region, we next see the action of local pre-processing operations in activating the hidden DCC of the mixed states.  -- do not understand trade off. }

%In fact for $N \to \infty$, the ensemble only contains the identity matrix.  In fact for N = 100, we obtain $\langle S(\rho) \rangle = 0.99$. If we use this result, then our analytical calculations for $\langle DC \rangle $ matches pretty well with the numerical average as shown in the table. The numerical results, of course, are independent of the choice of N and have been repeated in each case for comparison. \\

%We find that with N = 100, the analytical and numerical results are in excellent correspondence. However, there is a trade off Increasing N, decreases the number of distinct reduced density matrices in the random ensemble. In fact for $N \to \infty$, the ensemble only contains the identity matrix. For all the cases, however, the capacity of rank 2,3 and 4 states is the classical capacity = 1, so there are no issues on that count.

\section{Effects of local pre-processing on the dense coding capacity of Haar uniformly generated states}
\label{sec:DCCafter}

The dense coding capacity, given in Eqs. (\ref{DCC_formula1}) and (\ref{DCC_formula2}), are obtained by optimizing  over the unitary encoding performed by the sender(s) and the decoding by the receiver(s).  However, it is expected that before starting the DC protocol,  if one includes preprocessing on the shared states between the sender(s) and the receiver(s), the capacity can, in general, be enhanced with a certain probability. Since we deal with random states, and our aim is to find out the effects of preprocessing on random states, we illustrate by analyzing the situation where some of the senders and receivers or all of them  apply the local dichotomic POVMs ( in Eq. (\ref{POVM_i}) ) to activate the hidden DCC (to enhance DCC) when a particular choice of outcomes occur. 
%We show that for some states, POVMs can enhance DCC well beyond the classical limit. 
To that end, we try to derive analytical conditions, which when satisfied, ensure that the state can exhibit enhanced DCC after pre-processing by POVMs.
% If the DCC of the shared state does not show any quantum advantage and after pre-processing, it is useful for quantum DC,  we say that the state possesses $\textit{hidden}$ DCC. On the other hand, if initially the DCC is in the quantum region and after pre-processing, it increases, we call this increment in DCC as $\textit{enhanced}$ DCC. 
 %Note that if for a particular outcome of POVM, the post-processed state possess DCC, we discard the state (follow the best classical protocol). 
 Let us define the following figures of merit to monitor the action of pre-processing operations on DCC.\\

\textit{Optimal increase in dense coding capacity  (via POVM).} After maximizing over all the parameters involved in local POVM, we concentrate on the DCC of the resulting state which is obtained when a specific measurement outcome clicks. 
The  maximization is performed when POVM is performed by the sender(s) or the receiver(s) or both. We define the optimal increase in DCC due to the action of POVM by all the parties as
\begin{equation}
	\mathbb{O}_{DCC}= \underset{\{E_i^{o_i}\}}{\text{max }} \mathbb{C}\left( \frac{(\otimes \sqrt{E_{i}^{o_i}}) \rho (\otimes \sqrt{E_{i}^{o_i^{\dagger}}})}{\text{tr}\left[(\otimes \sqrt{E_{i}^{o_i}}) \rho (\otimes \sqrt{E_{i}^{o_i^{\dagger}}})\right] }\right), 
	%= \mathbb{C}(\rho_o)
	\label{Optimumcost_DCC}
\end{equation}
where the numerator denotes the output state, $\text{tr}\left[(\otimes \sqrt{E_{i}^{o_i}}) \rho (\otimes \sqrt{E_{i}^{o_i^{\dagger}}})\right]$ is the probability of obtaining the outcome to normalize the state,  and\(\{o_i\}\) represents the particular outcome that gives the maximal DCC. In case, some of the parties perform POVMs, we apply the identity operator on the rest as mentioned in Eq. (\ref{POVM_state1}). 
%and $\rho_o$ is the normalised state after the action of pre-processing.
Although it may occur that several sets maximize the capacity, in a realistic situation, POVM is set to the  optimal direction so that  one of the possible choices of outcomes can  occur. The enhancement can be measured by evaluating \(\mathbb{O}_{DCC} - \mathbb{C}(\rho)\), with \(\mathbb{C}(\rho)\) being the DCC of the original state before application of POVM. 
%So, we look at the average action of the pre-processing. We propose the following two cost function to study the average action of pre-processing applicable in realistic situation.\\

\textit{Average cost of optimum dense coding capacity.}  Let us suppose that an outcome, \(o_i\), of a particular POVM gives the maximal enhancement in the capacity of dense coding. The average cost of optimum dense coding capacity is then defined as 
%we take the maximum DCC corresponding to particular choice of outcome multiplied by its probability of occurrence along with this we take the DCC for other outcome choices multiplied by their respective probability of occurrence. Note that the DCC for other outcome choices are calculated at the same value of POVM parameters at which maximum DCC occurred. If for a particular choice of outcome, DCC after pre-processing is in the classical region, we take the maximum classical DCC in that case.  Let $\{o_i\}$ is the set of all possible outcomes. The average dense coding capacity cost-1 is  defined as:
\begin{equation}
	\mathbb{A}_{DCC}^{1} = \sum_{\{o_i\}} p_{o_i} \mathbb{C}(\rho_{o_i}),
	% \quad o_i \in \{o_i\},
	\label{averagecost1_DCC}
\end{equation}
where $p_{o_i}$ is the probability of occurrence of a particular outcome of POVM, \(o_i\) and $\mathbb{C} (\rho_{o_i})$ represents the DCC of the normalized state after the action of POVM for that particular outcome, \(o_i\).  Note that the DCC for other outcome choices is calculated with the same choices of parameters in POVM  which leads to the maximum increase in DCC calculated in Eq. (\ref{Optimumcost_DCC}). 

\textit{Maximal cost of average dense coding capacity.} 
After performing local POVMs by the  sender(s) and the receiver(s), if we are interested to know the maximum enhancement that can occur in the dense coding protocol on average, we can evaluate the quantity, given  by 
%
% In this case,  we calculate the sum of DCC corresponding to all outcome choices multiplied by their respective probability of occurrence. We then maximize this expression over the allowed range of POVM parameters. The average dense coding capacity cost-2 is  defined as:
\begin{equation}
	\mathbb{A}_{DCC}^{2} = \underset{\{\lambda_i, \theta_i, \phi_i\}}{\text{max}} \left[ \sum_{o_i} p_{o_i} \mathbb{C}(\rho_{o_i}) \right], 
	\label{averagecost2_DCC}
\end{equation}
where the maximization is performed over the set of  parameters in POVM as given in Eq. (\ref{POVM_i}), $p_{o_i}$ is the probability of occurrence of a particular outcome and $\rho_{o_i}$ is the normalized state after the action of pre-processing for that particular outcome.
 Notice that in the case of average cost of optimum dense coding capacity, we perform maximization to identify a single outcome that gives maximum DCC after POVM while in this case, maximization is performed to optimize the entire quantity which is written in the box parenthesis. 

Based on the above three quantities, we now analyze the consequence of pre-processing acted by different combinations of the sender(s) and the receiver(s) mentioned before on the dense coding. Similar quantities will also be considered for teleportation where capacities will be replaced by fidelities.

\subsubsection{Random two-qubit states after POVM: A single sender and a single receiver scenario}

In this scenario, a sender, $S$, and a receiver, $R$, share a two-qubit random state, \(\rho^{SR}\) having different ranks. When the shared state is pure, we know that whenever the state is entangled, it is dense codable and hence the hidden DCC cannot be revealed after POVM although POVM can enhance the dense coding capability of the shared pure state.  On the other hand, if the shared state is a two-qubit mixed state, we find that the mean DCC is below unity, as shown in Sec. \ref{sec:mean}, thereby implying that most of the Haar uniformly generated states do not show a quantum advantage in DC. For these states, either the sender or the receiver or both of them apply the pre-processing operations to extract the hidden DCC. We now present the exact conditions (in terms of eigenvalues of the shared state and reduced state before and after preprocessing) which have to be satisfied by the rank-2 mixed states  for extracting the hidden DCC. It is important to mention here that when the pre-processing is  completely positive trace preserving  (CPTP) map which can be included in the encoding-decoding process, it was shown  that DCC can be enhanced by applying CPTP operations neither by the sender nor by the receiver \cite{Horodecki_2012}. 
%Moreover, the conditions derived below does not depend on the POVM. 

\textit{Rank-2  state in 1S-1R scenario.} Let the shared state, $\rho^{SR}$, be a rank-2 state. We denote its eigenvalues by $x_1$ and $x_2$, with $x_1 + x_2 = 1$ and $x_2 - x_1 = k_0$ while the reduced state at the receiver's side, $\rho^{R} = \text{tr}_{S}(\rho^{SR})$, have eigenvalues $x'_1$ and $x'_2$ whose sum is still unity and difference is taken as $x'_2 - x'_1 = k'_0$. Obviously, since all the eigenvalues are positive, both $k_0, k'_0 <1$. The DCC (before pre-processing) expression can then easily be written in terms of the sum and difference of the eigenvalues of the density matrices as 
%\begin{widetext}
	\begin{eqnarray}
	&&\mathbb{C}(\rho^{SR}) =  1 \nonumber \\
	% \nonumber 1 + S(\rho^{R}) - S(\rho^{SR})  =  
	%1 - (x'_1\log_2(x'_1) + x'_2\log_2(x'_2)) + (x_1\log_2(x_1) + x_2\log_2(x_2)) \\ 
	&& 	- ((1 - k'_0)\log_2(1 - k'_0)	+  (1 + k'_0)\log_2(1 + k'_0)) \nonumber \\
	&&+ ((1 - k_0)\log_2(1 - k_0) + (1 + k_0)\log_2(1 + k_0)).
	\end{eqnarray}
%\end{widetext}
After pre-processing  has been applied, the resulting state is $\rho_p^{SR}$ with eigenvalues summing to unity and having $k$ as their difference. Similarly, for the reduced state after pre-processing ($\rho_p^{R}$), the sum of eigenvalues is unity, but their difference is $k'$ 
%We must note that since the eigenvalues are all still positive semi-definite, 
and $0 \leq k,k' \leq 1$. In a similar spirit as above, we can write the dense coding capacity after pre-processing in terms of $k$ and $k'$ as
%\begin{widetext}
	\begin{eqnarray}
	\label{eq:afterpre}
	&& \mathbb{C}(\rho_p^{SR})  =   1 + S(\rho_p^{R}) - S(\rho_p^{SR})  = 1\nonumber \\
	&&- ((1 - k')\log_2(1 - k')	+  (1 + k')\log_2(1 + k')) \nonumber \\
	&&+ ((1 - k)\log_2(1 - k) + (1 + k)\log_2(1 + k)).
	\end{eqnarray}
%\end{widetext}
It is straightforward to show that each entropy term, \(S(\rho^R)\), \(S(\rho_p^R)\), \(S(\rho^{SR})\), and \(S(\rho_p^{SR})\),     in both $\mathbb{C}(\rho^{SR})$ and $\mathbb{C}(\rho_p^{SR})$, reaches their individual  maximum values when $k_0,k'_0,k, k' $ are all vanishing. 
%With this in mind, we derive conditions for enhancement of DCC after pre-processing in terms of the above parameters. 
The DCC after pre-processing should possess two properties --  the DCC after pre-processing is in the quantum region, i.e., $\mathbb{C}(\rho_p^{SR}) > 1$; and the DCC after pre-processing is greater than that of before, i.e., $\mathbb{C}(\rho_p^{SR}) > \mathbb{C}(\rho^{SR})$.
%\end{enumerate}
We derive conditions for both these traits and argue whether both are necessary for a given rank or if we can work with either of them. 

\emph{ Condition for non-classical DCC after pre-processing.} This condition demands that $S(\rho_p^{R}) - S(\rho_p^{SR}) > 0$. Since both these terms increase when their  differences, i.e., $k,\, k'$ approach to zero, and since after pre-processing we find numerically that the differences decrease i.e. $k<k_0$ and $k'<k'_0$, we propose the following: \\

\textbf{Proposition 1.} \textit{The dense coding capacity after pre-processing is non-classical i.e. $S(\rho_p^R) - S(\rho_p^{SR}) > 0$, when $k'$ is smaller than $k$, i.e. $k'<k$.} \\

The above condition follows from Eq. (\ref{eq:afterpre}) and  guarantees that $\mathbb{C}(\rho_p^{SR}) >1$ although it does not ensure enhancement after pre-processing. States which satisfy it after pre-processing, will surely have non-classical DCC and vice-versa. We see that this condition involves two eigenvalues and it may seem that it holds only for rank-2 states. However, our numerical analysis suggests that some rank-3 and  -4 states after optimal POVMs are reduced to states with rank-2, and hence this condition is true for such two-qubit mixed states as well.

\emph{ Condition for enhancing DCC after pre-processing. } This condition demands that $S(\rho_p^{R}) - S(\rho_p^{SR}) > S(\rho^{R}) - S(\rho^{SR})$ which in turn 
%Keeping the limits on the difference parameters for maximization of the entropy terms in mind, 
%we see that the satisfaction of the inequality 
depends on the changes occurred in the difference of eigenvalues  before and after the pre-processing. In particular, we observe the following after POVM:
%This can be seen easily for rank 2 states, but for higher ranks, it is non trivial due to the presence of more than two eigenvalues. We now make our second proposition as follows:

\textbf{Proposition 2.} \textit{The dense coding capacity of rank-2  two-qubit states after pre-processing is greater than that of the state without pre-processing, if %the difference between  for the reduced subsystem goes significantly closer to zero than the difference parameter of the shared total state, or in other words,
 $(k'_0 - k')>(k_0 - k)$}.\\

As noticed, when the difference between eigenvalues of the states  vanishes,  the  individual entropies are maximized Therefore, we can get enhancement after pre-processing,  if  the rate in which $k'$ in $\rho_p^{R}$ goes closer to zero after changing from $k'_0$ to $k'$  is higher than $k$ in  $\rho_p^{SR}$ which changes from $k_0$ to $k$, then the increase in $S(\rho_p^{R})$ (from $S(\rho^{R})$) is greater than the increase in $S(\rho_p^{SR}$ (from $S(\rho^{SR})$), thereby implying $\mathbb{C}(\rho_p^{SR}) > \mathbb{C}(\rho^{SR})$. We observe that among randomly generated two-qubit rank-2 states, $80.04\%$ states to satisfy the above condition, although there are states, showing the advantage of pre-processing which do not satisfy the above criteria. 

%We need to be careful with this condition. $DC' > DC_0$ implies proposition 2, but not always the other way around. Even if proposition 2 is satisfied, we may have $DC' = DC_0$ and hence, no enhancement. To confirm this, we make the condition stricter by saying that we need $(k'_0 - k') - (k_0 - k) \geq 0.05$. This will guarantee $DC' \geq DC_0$. However, there may be states which exhibit the desired criterion without satisfying this stricter inequality. Numerically, we have seen that this condition picks out 80\% of all states which have $DC' > DC_0$.\\

%The enhancement of DCC after POVM is shown in figure (\ref{dcr2}).
%\begin{figure}[!ht]
%	\resizebox{9cm}{8cm}{\includegraphics{DCr2.pdf}} 
%	\caption{\footnotesize (Coloronline) Non-classical dense coding capacity of randomly generated rank 2 states is plotted before pre-processing (in red solid) and after POVM (in blue dashed).}
%	\label{dcr2}
%\end{figure}

%After POVM optimization, if we take the weighted average over all the possible outcomes, then the fraction of dense codeablew states does not increase. However, there is a surge in the number of states with higher non-classical dense codeability $(\geq 1)$ at the expense of those with lower values.

%\begin{figure}[!ht]
%	\resizebox{9cm}{8cm}{\includegraphics{r2_dc_hist.pdf}}
%	\caption{\footnotesize (Coloronline) [Rank 2] Normalised number of states is plotted against dense coding capacity before (in blue blank) and after (in red checks) preprocessing}
	
%\end{figure}
\begin{figure}[!ht]
	\resizebox{9cm}{8cm}{\includegraphics{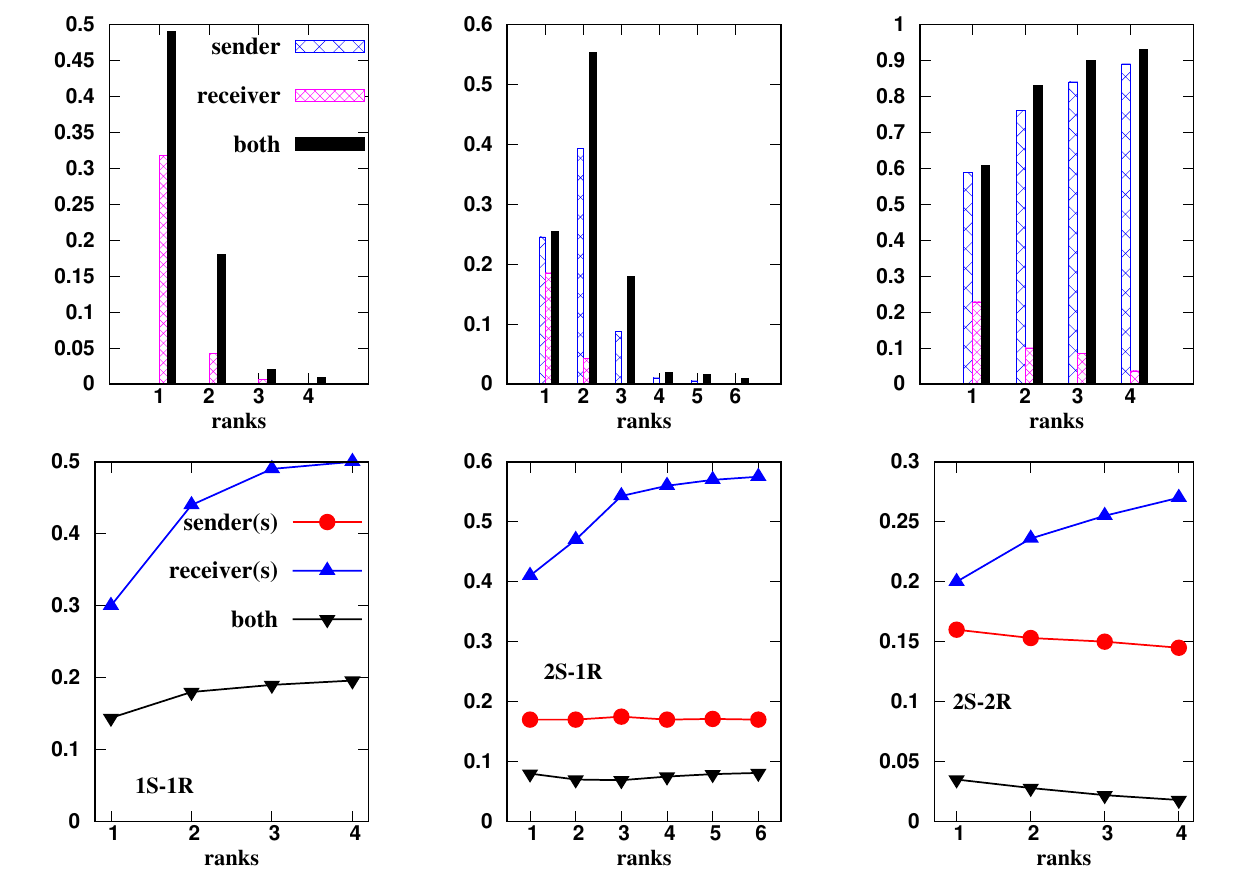}}
	\caption{ (Color online)  Upper Panel. Mean optimal increase in dense coding capacity, \(\overline{\mathbb{O}}_{DCC}\),  (ordinate)  in 1S-1R, 2S-1R and 2S-2R scenarios with varying rank of the shared state (abscissa).  Lower Panel. Average probability of obtaining optimal increase in DCC, \(\overline{p}_{\mathbb{O}_{DCC}}\), in  1S-1R, 2S-1R, and 2S-2R DC protocol  with ranks. All the axes are dimensionless. }
\label{fig:maxprobDCC}
\end{figure}

Let us now move to the scenario where two-qubit Haar uniformly generated states undergo pre-processing and we first address the issue of activation of hidden DCC, with the increase of ranks and with the number of parties doing POVM. 
%\begin{widetext}

\begin{table}[]
	\resizebox{0.5\textwidth}{!}{\begin{minipage}{0.7\textwidth}
			\caption{Average optimal increase in DCC (Two-qubits). $\overline{\mathbb{O}}_{DCC}$   denotes the optimal increase in DCC on average. "Both" and  "Receiver"  indicate that POVM is performed both by the sender, as well as the receiver, and the receiver only. ``Before'' represents the mean DCC for a given rank without local POVM.}
			\label{2DCCmax}
			\centering
			\begin{tabular}{|l|l|l|l|}
				\hline
				& \multicolumn{1}{c|}{Before} & \multicolumn{1}{c|}{Receiver}                                                                                & \multicolumn{1}{c|}{Both}                                                                                    \\ \hline
				& \multicolumn{1}{c|}{}       & \multicolumn{1}{c|}{}                                                                                        & \multicolumn{1}{c|}{}                                                                                        \\ \hline
				& \multicolumn{1}{c|}{}       & \multicolumn{1}{c|}{$\overline{\mathbb{O}}_{DCC}$} & \multicolumn{1}{c|}{$\overline{\mathbb{O}}_{DCC}$} \\ \hline
				&                             &                                                                                                              &                                                                                                              \\ \hline
				Rank-1 & 1.48                        & 1.8                                                                                                          & 1.97                                                                                                         \\ \hline
				Rank-2 & 1.07                        & 1.11                                                                                                         & 1.25                                                                                                         \\ \hline
				Rank-3 & 1.0043                      & 1.01                                                                                                         & 1.034                                                                                                        \\ \hline
				Rank-4 & 1.00026                     & 1.002                                                                                                        & 1.01                                                                                                         \\ \hline
			\end{tabular}
	\end{minipage}}
\end{table}

%\end{widetext}

\begin{enumerate}
\item \textit{Effects of rank.} As depicted in Fig. \ref{fig:distriDCC},  the DCC of  most of the mixed random states lies just above the classical limit if they have non-classical DCC and the percentage of states that have non-classical DCC decreases sharply with increasing rank. 

First of all, we notice that if POVMs are performed by the sender, no increment in DCC is observed for two-qubit states (cf. \cite{Horodecki_2012}). \\

Secondly,  after POVM, the \emph{optimal increase in DCC}  shows a rise on average (see Fig. \ref{fig:maxprobDCC}), albeit with a finite probability. In Table \ref{2DCCmax} and Fig. \ref{fig:maxprobDCC}, we illustrate \(\overline{\mathbb{O}}_{DCC} = \frac{\sum \mathbb{O}_{DCC} (\rho^{SR})}{N_S}\) and  \(\overline{p}_{\mathbb{O}_{DCC}} = \frac{\sum p_{\mathbb{O}_{DCC}} (\rho^{SR})}{N_S}\), where \(N_S\) is the total number of states simulated and \(p_{\mathbb{O}_{DCC}} (\rho^{SR})\) is the probability of obtaining the outcome of the POVM which leads to the state having maximum increase in DCC.  The increment and corresponding probability are complimentary to each other i.e. more increment occurs with lesser probability as it is visible from the upper and lower panels of Fig. \ref{fig:maxprobDCC}. It is true that since most of the rank-4 and above randomly generated states without pre-processing are not advantageous for quantum DC, after pre-processing, the increase is also very low on average. 

\begin{figure}[!ht]
	\resizebox{9cm}{8cm}{\includegraphics{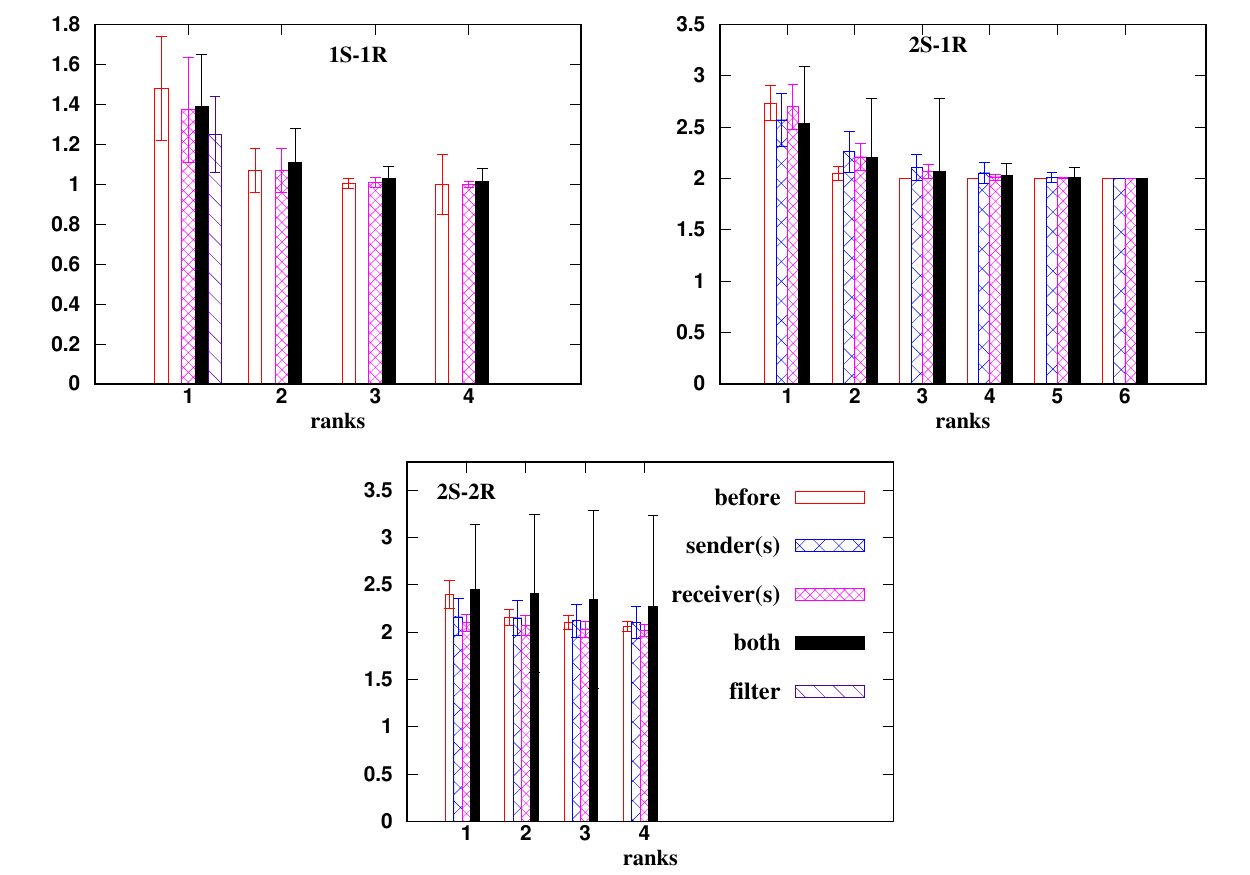}}
	\caption{\footnotesize (Color online)  Upper Panel. The mean and standard deviation (shown as error bars) of average cost of optimum dense coding capacity, \(\overline{\mathbb{A}}_{DCC}^{1}\), for randomly generated states  against the rank of the state in 1S-1R, 2S-1R. Lower Panel. Same quantity  is plotted for two senders -two receivers. All the axes are dimensionless.}
\label{fig:meanplotDCC}
\end{figure}

To analyse the average cost of optimum dense coding capacity and maximal cost of average DCC, we evaluate the mean and the standard deviation (SD) of these quantities for randomly generated two-qubit states. We observe that although POVMs  by the receiver or  both by the sender and the receiver do not help to increase the mean and the SD of these quantities for pure states, the pre-processing indeed enhances the capability of showing quantum advantages in the dense coding protocol in  states with rank-2 and above as shown in the left columns of the upper panel in Figs.  \ref{fig:meanplotDCC} and \ref{fig:newmeanplotDCC}  as well as in Table \ref{meanDCCtab}. When both the parties apply local POVMs, we observe that SD of \(\mathbb{A}_{DCC}^{1}\) decreases with rank also although the SD obtained from the frequency distribution of DCC before pre-processing is lower than that of quantities after POVM. 
%
%
%\textit{Hidden} DCC is revealed for higher ranked states. The average DCC cost-1 and 2 also exhibit enhancement for rank-3 and 4 states as shown in Figs.  (\ref{meanplotDCC}) and (\ref{newmeanplotDCC}). Hence, our scheme activates the hidden DCC in realistic situation also. 

\begin{figure}[!ht]
	\resizebox{9cm}{8cm}{\includegraphics{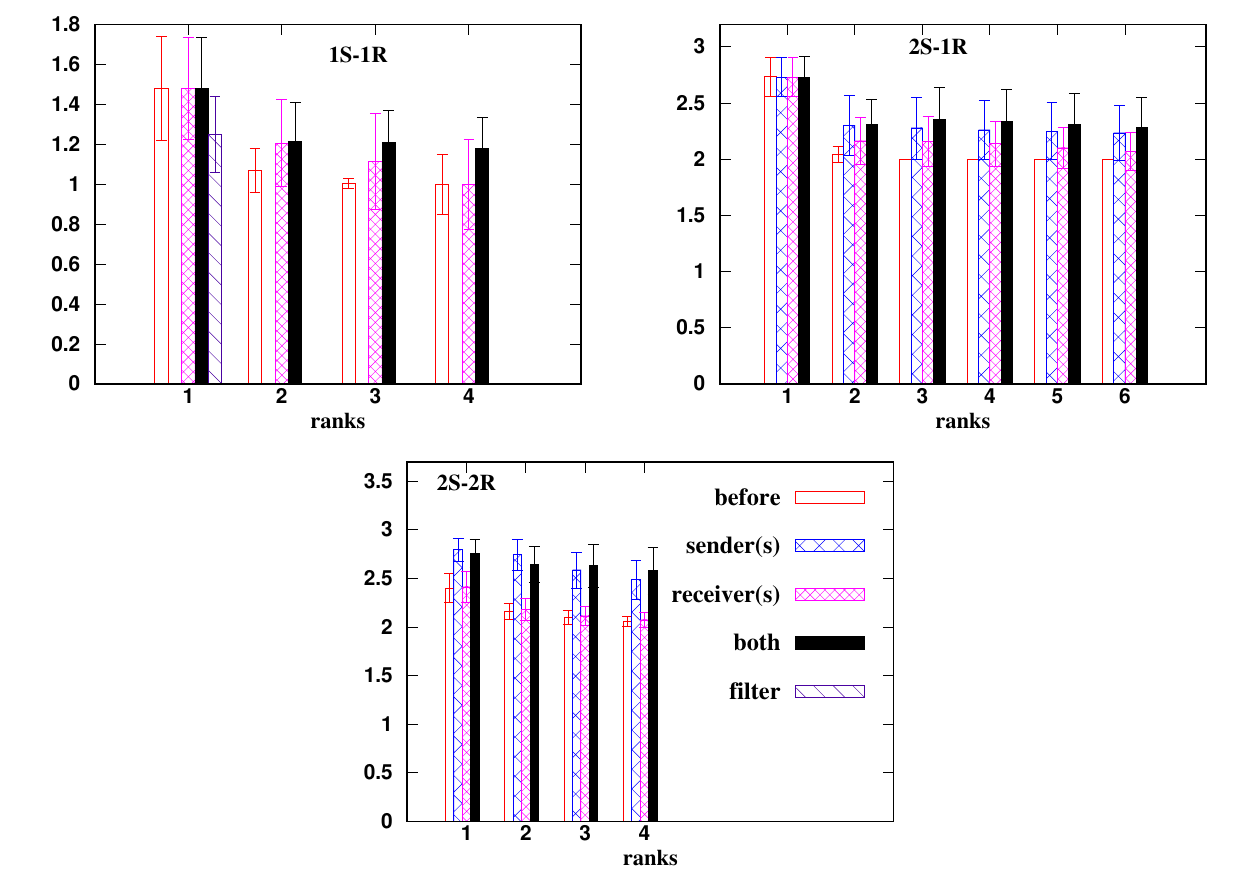}}
	\caption{\footnotesize (Color online)  Upper Panel. The mean of maximal cost of average DCC, \(\overline{\mathbb{A}}_{DCC}^{2}\),  against  ranks.
	All the other specifications are same as in Fig. \ref{fig:meanplotDCC}. }
\label{fig:newmeanplotDCC}
\end{figure}

\item \textit{Effect of number of party doing pre-processing:} As mentioned before, in the two-qubit scenario, no POVMs by the sender  enhances the DCC while the receiver's POVM help. However, when both parties apply POVMs, the enhancement is more pronounced than the case when only the receiver acts which can be confirmed by all the figures of merits considered here to measure the performance of DC in this case. 

%The enhancement is more when all parties apply POVMs (\ref{newmeanplotDCC}). The average DCC costs shows no overall change when only sender or only receiver is doing pre-processing but it increases when both of them do pre-processing. We also noted that although the enhancement in DCC in single sided pre-processing is less but the success probability is high.
\end{enumerate}

\begin{table}[]
	\resizebox{0.5\textwidth}{!}{\begin{minipage}{0.6\textwidth}
			\caption{ Mean of average cost of optimum dense coding capacity, $\overline{\mathbb{A}}_{DCC}^{1}$ and maximal cost of average dense coding capacity, $\overline{\mathbb{A}}_{DCC}^{2}$ for two-qubit random states. The label of the column,  ``Receiver'' and ``Both'' indicate respectively the quantities after POVMs  are applied by the receiver and after both the parties perform  POVMs.}
			% \textcolor{red}{both and receiver have to be interchanged, no? Include SD.} } 
			\label{meanDCCtab}
			\centering
		\begin{tabular}{|l|l|l|l|l|l|}
			\hline
			& \multicolumn{1}{c|}{Before} & \multicolumn{2}{c|}{Receiver}                                                                                                                                                                                             & \multicolumn{2}{c|}{Both}                                                                                                                                                                                                 \\ \hline
			& \multicolumn{1}{c|}{}       & \multicolumn{1}{c|}{}                                                                                        & \multicolumn{1}{c|}{}                                                                                      & \multicolumn{1}{c|}{}                                                                                        & \multicolumn{1}{c|}{}                                                                                      \\ \hline
			& \multicolumn{1}{c|}{}       & \multicolumn{1}{c|}{$\overline{\mathbb{A}}_{DCC}^{1}$} & \multicolumn{1}{c|}{$\overline{\mathbb{A}}_{DCC}^{2}$} & \multicolumn{1}{c|}{$\overline{\mathbb{A}}_{DCC}^{1}$} & \multicolumn{1}{c|}{$\overline{\mathbb{A}}_{DCC}^{2}$} \\ \hline
			&                             &                                                                                                              &                                                                                                            &                                                                                                              &                                                                                                            \\ \hline
			Rank-1 & 1.48                        & 1.373                                                                                                        & 1.48                                                                                                       & 1.39                                                                                                         & 1.48                                                                                                       \\ \hline
			Rank-2 & 1.07                        & 1.07                                                                                                         & 1.21                                                                                                       & 1.11                                                                                                         & 1.22                                                                                                       \\ \hline
			Rank-3 & 1.0043                      & 1.01                                                                                                         & 1.15                                                                                                       & 1.03                                                                                                         & 1.21                                                                                                       \\ \hline
			Rank-4 & 1.00026                     & 1                                                                                                            & 1.02                                                                                                       & 1.013                                                                                                        & 1.2                                                                                                        \\ \hline
		\end{tabular}
	\end{minipage}}
\end{table}

%\begin{figure}[!ht]
%	\resizebox{9cm}{8cm}{\includegraphics{r3_pmax_hist.pdf}}
%	\caption{\footnotesize (Coloronline) [Rank 3] Normalized number of states is plotted against dense coding capacity before (in blue checks) and after (in red blank) preprocessing}	
%	\end{figure}

\begin{table}[]
	\resizebox{0.5\textwidth}{!}{\begin{minipage}{0.7\textwidth}
			\caption{Percentage of non-classical dense coding capacity for three-qubit states with rank-1 to rank-6 before and after POVMs. All the notations are the same as in Table \ref{meanDCCtab}.  Again ``Before'' denotes the percentage of states, giving a quantum advantage in the  2S-1R DC protocol without pre-processing. } 
			\label{3DCCperc}
			\centering
			\begin{tabular}{|l|r|r|r|r|r|r|r|}
				\hline
				\multicolumn{1}{|c|}{} & \multicolumn{1}{c|}{Before} & \multicolumn{2}{c|}{Senders}                                                                                   & \multicolumn{2}{c|}{Receiver}                                                                                  & \multicolumn{2}{c|}{Both}                                                                                      \\ \hline
				\multicolumn{1}{|c|}{} & \multicolumn{1}{c|}{}       & \multicolumn{1}{c|}{}                                  & \multicolumn{1}{c|}{}                                 & \multicolumn{1}{c|}{}                                  & \multicolumn{1}{c|}{}                                 & \multicolumn{1}{c|}{}                                  & \multicolumn{1}{c|}{}                                 \\ \hline
				\multicolumn{1}{|c|}{} & \multicolumn{1}{c|}{}       & \multicolumn{1}{c|}{$\overline{\mathbb{A}_{DCC}^{1}}$} & \multicolumn{1}{c|}{$\overline{\mathbb{A}_{DCC}^{2}}$} & \multicolumn{1}{c|}{$\overline{\mathbb{A}_{DCC}^{1}}$} & \multicolumn{1}{c|}{$\overline{\mathbb{A}_{DCC}^{2}}$} & \multicolumn{1}{c|}{$\overline{\mathbb{A}_{DCC}^{1}}$} & \multicolumn{1}{c|}{$\overline{\mathbb{A}_{DCC}^{2}}$} \\ \hline
		& \multicolumn{1}{l|}{}       & \multicolumn{1}{l|}{}                                  & \multicolumn{1}{l|}{}                                 & \multicolumn{1}{l|}{}                                  & \multicolumn{1}{l|}{}                                 & \multicolumn{1}{l|}{}                                  & \multicolumn{1}{l|}{}                                 \\ \hline
	Rank-1                 & 100\%                       & 95.55\%                                                & 99.96\%                                               & 98\%                                                   & 99.99\%                                               & 87.82\%                                                & 99.83\%                                               \\ \hline
Rank-2                 & 50.31\%                     & 86.72\%                                                & 92.35\%                                               & 92.28\%                                                & 91.25\%                                               & 45.83\%                                                & 96.05\%                                               \\ \hline
Rank-3                 & 0.08\%                      & 77.33\%                                                & 79.64\%                                               & 77.47\%                                                & 78.45\%                                               & 36.26\%                                                & 87.36\%                                               \\ \hline
Rank-4                 & 0\%                         & 58.1\%                                                 & 75.26\%                                               & 38.91\%                                                & 50.26\%                                               & 30.6\%                                                 & 86.84\%                                               \\ \hline
Rank-5                 & 0\%                         & 49.07\%                                                & 73.96\%                                               & 10.21\%                                                & 48.85\%                                               & 25.32\%                                                & 85.28\%                                               \\ \hline
Rank-6                 & 0\%                         & 39.78\%                                                & 71.23\%                                               & 7.26\%                                                 & 45.21\%                                               & 18.96\%                                                & 82.44\%                                               \\ \hline
\end{tabular}
	\end{minipage}}
\end{table}

\subsubsection{Local POVMs by two senders are more effective than a single receiver}

Three-qubit Haar uniformly generated states with rank-1 to  rank- 6 shared between two senders, $S_1$ and $S_2$ and a single receiver $R$ are considered. 
All three-qubit pure random states show a quantum advantage in DC since the random states are typically genuinely multiparty entangled and hence \(S(\rho^R)\) is positive for all of  them. With the increase of rank, states showing non-classical DCC decreases and we do not find a single randomly generated state having rank$\geq $4 which has \(\mathbb{C}^{2S-1R} (\rho^{S_1S_2R}) > 2\) as shown in Table \ref{3DCCperc}. Unlike two-qubit states, we observe that local POVMs applied by the senders  can  also help to enhance DCC probabilistically (see Fig. \ref{fig:maxprobDCC}). Figs. \ref{fig:meanplotDCC} and \ref{fig:newmeanplotDCC} depict the enhancement on average by considering \(\overline{\mathbb{A}}_{DCC}^1\) and \(\overline{\mathbb{A}}_{DCC}^2\) due to the application of local POVMs before starting the protocol.  
% As before, we find that the mixed states of higher rank have mostly classical capacity (\ref{3DCCperc}). The trends in enhancement of DCC after POVM application are similar to the previous case, with POVMs by all parties providing the best results (\ref{fig:newmeanplotDCC}).
In stark contrast with the two-qubit case, we observe that if senders can apply local POVMs, the maximal cost of average dense coding capacity gets more increased compared to the situation when only receiver performs POVM. 
Moreover,  our results demonstrate that to obtain a quantum advantage in a multipartite DC scheme for random density matrices, pre-processing is essential.  
% a substantial increase in mean DCC of rank-3 and rank-4 states after pre-processing by all parties ($\langle \mathbb{C} \rangle$ before is 2 and after is 2.4). 

\subsubsection{Effects of POVM on the upper bound of DCC with 2S-2R case}

Since  two senders- two receivers DC scenario, only  upper bound is known, we will now see whether upper bound can be enhanced by using  pre-processing on the shared states. It is interesting to note here that there are states for which the upper bound on DCC by LOCC can be saturated.   All the four-qubit pure states which are, in general, genuinely multipartite entangled states show \(\mathbb{C}^{2S-2R} (\rho^{S_1 S_2 R_1 R_2}) 
\leq U^{2S-2R} ( >2 )\). 
% Like two- and three-qubit pure states, 
%Although POVMs also reduce the percentages of pure  states to show quantum advantage.
 Interestingly,  we observe that  \(\overline{\mathbb{A}}_{DCC}^i,\, i=1, 2\) increases after applying optimal POVMs by both the parties even for pure states which is not true for DC protocol involving a single receiver. 
As seen from Figs. \ref{fig:maxprobDCC}, \ref{fig:meanplotDCC} and \ref{fig:newmeanplotDCC} and Table \ref{4DCC},  
for rank-2 to rank-4 four-qubit Haar uniformly generated states,  the upper bound can again be improved substantially if the parties perform local POVM.  Like DC with  \(2S-1R\) scenario, senders can increase the upper bound more  by acting POVMs  compared to the case when receivers apply local POVMs which is prominent for \(\overline{\mathbb{A}}_{DCC}^2\). 
%Such observations possibly indicate that the increment of DCC by pre-processing does not linearly depend only on  the number of senders anor the number of receivers. 

%In this scenario, two senders, $S_1$ and $S_2$ along with two receivers $R_1$ and $R_2$ share a four qubit random state. First two qubits are on senders side whereas the receivers have the last two qubits. The DCC is above the classical capacity for low ranked states initially but not optimum. Upon pre-processing DCC approaches the optimum value for some states. The trends in the enhancement are similar to single sender-single receiver scenario, see \ref{4DCC} and the complementarity between optimum DCC and its corresponding probability exists in this case as well.(\ref{fig:distriDCC}). A notable difference that we observed is that in this case there is an advantage if only senders do the pre-processing (\ref{fig:newmeanplotDCC}). The initial mean DCC for rank-1 and rank-2 state is {\color{blue} 2.42 and 2.21} respectively. The increment in mean DCC is {\color{blue} 0.38 and 0.52} for rank-1 and rank-2 states respectively when only senders are doing pre-processing where it is {\color{blue} 0.35 and 0.45} respectively when all parties are doing pre-processing. This clearly demonstrate the advantage of single-sided pre-processing by senders only over both sided pre-processing.

\begin{table}[]
	\resizebox{0.5\textwidth}{!}{\begin{minipage}{0.7\textwidth}
			\caption{\(\overline{\mathbb{A}}_{DCC}^i,\, i=1, 2\) are listed for four-qubit states performing a DC protocol with two senders and two receivers. %\textcolor{red}{If it is difficult to adjust SD here, we can put the no in Figs. } 
			}  
			\label{4DCC}
			\centering
			\begin{tabular}{|l|l|l|l|l|l|l|l|}
				\hline
				& Before & \multicolumn{2}{l|}{Senders}                                          & \multicolumn{2}{l|}{Receivers}                                        & \multicolumn{2}{l|}{Both}                                             \\ \hline
				&        &                                   &                                   &                                   &                                   &                                   &                                   \\ \hline
				&        & $\overline{\mathbb{A}_{DCC}^{1}}$ & $\overline{\mathbb{A}_{DCC}^{2}}$ & $\overline{\mathbb{A}_{DCC}^{1}}$ & $\overline{\mathbb{A}_{DCC}^{2}}$ & $\overline{\mathbb{A}_{DCC}^{1}}$ & $\overline{\mathbb{A}_{DCC}^{2}}$ \\ \hline
				&        &                                   &                                   &                                   &                                   &                                   &                                   \\ \hline
				Rank-1 & 2.4    & 2.16                              & 2.757                             & 2.1                               & 2.412                             & 2.45                              & 2.792                             \\ \hline
				Rank-2 & 2.16   & 2.15                              & 2.646                             & 2.07                              & 2.182                             & 2.41                              & 2.747                             \\ \hline
				Rank-3 & 2.1    & 2.12                              & 2.586                             & 2.032                             & 2.115                             & 2.35                              & 2.63                              \\ \hline
				Rank-4 & 2.06   & 2.1                               & 2.487                             & 2.021                             & 2.074                             & 2.27                              & 2.58                              \\ \hline
			\end{tabular}
	\end{minipage}}
\end{table}

\begin{figure}[!ht]
	\resizebox{8cm}{7cm}{\includegraphics{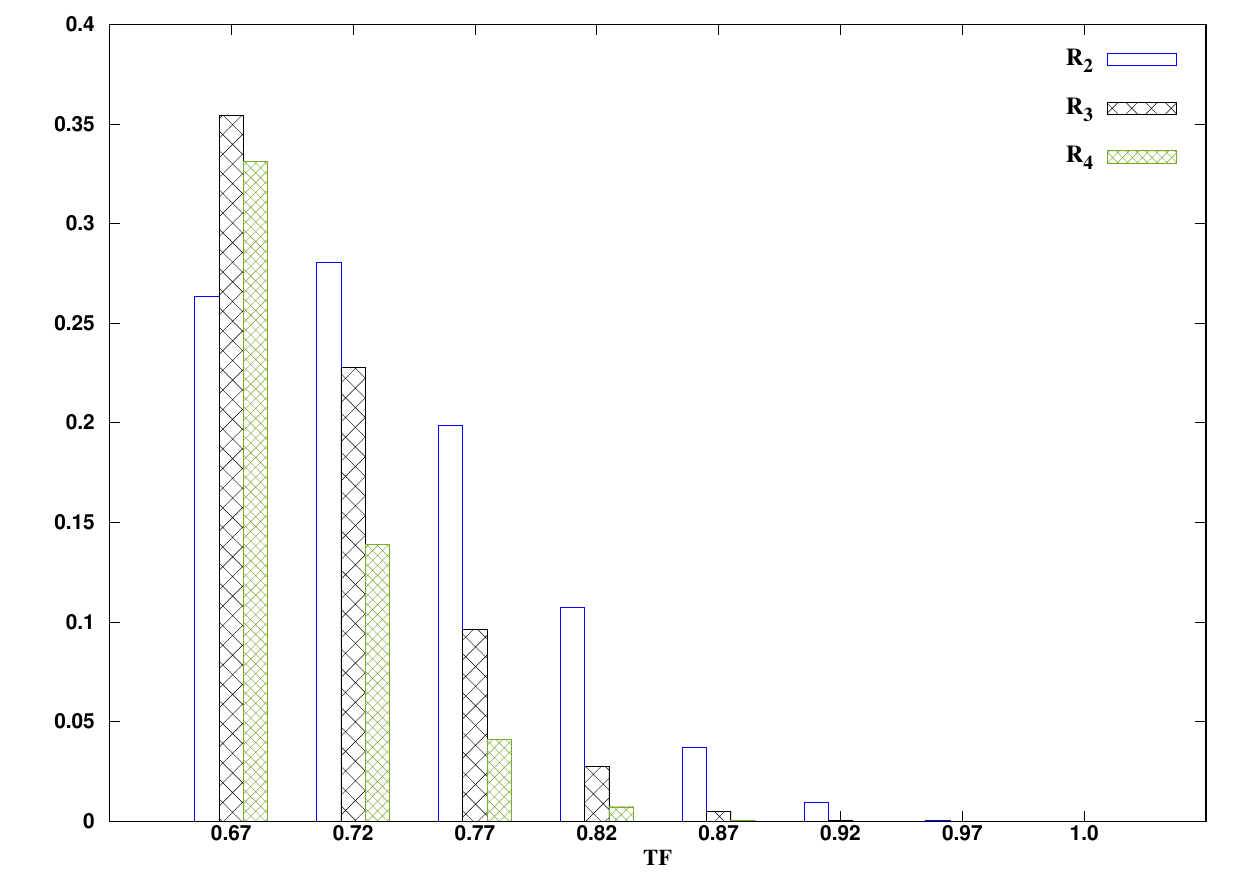}}
	\caption{ (Color online)  Normalized frequency distribution of TF, as defined in case of DC, for Haar uniformly generated two-qubit states (vertical axis) against  non-classical TF (horizontal axis).  All the axes are dimensionless.}
\label{before_hist_TF}
\end{figure}

\section{Teleportation Fidelity for Random States}
\label{TF}

Let us now move to another quantum communication protocol, in particular, quantum teleportation. Let us first analyze the frequency distribution of the teleportation  fidelity  for random two-qubit states with different  ranks in Fig. \ref{before_hist_TF}.   It was realized from different studies that higher the entanglement, higher is the TF of the two-qubit states and all pure two-qubit states are good for quantum teleportation as well as they violate Bell inequality. It was found \cite{Horo96} that  TF and violation of Bell inequality \cite{Bell, CHSH} are connected.
% We observed the similar behavior as shown in fig. (\ref{before_hist_TF}). We 
We observe that non classical TF for random states decreases with the  increase of the ranks of the states. For example, we find that  $48.2 \%$ rank-4 states have TF in the classical range while in rank-2 and rank-3, the percentages turn out to be \(10.14\%\) and \(20.91\%\) respectively.

Let us now show that with increasing rank, the relative number of states that possess local hidden variable model but gives non-classical fidelity increases.
 For example, $90 \%$ rank-2 states have  \(\mathbb{F} > 2/3\) out of which $67.9 \%$ are local while for rank-3 and -4, $93 \%$ and $98\%$ are local among  $73 \%$ and $51.8 \%$ states which show quantum advantage in teleportation respectively. Entire calculations and analysis are based on $10^6$ Haar uniformly generated states for each case. We demonstrate the action of local pre-processing operations in revealing the hidden TF of such states. In this regard, we later present the exact POVM operations that either one party or both the party has to apply on the shared pure random state to achieve optimum TF. 

\subsection{Effect of local pre-processing on teleportation fidelity}
\label{subsec:preprocessTF}

Like the DC protocol, either the sender or the receiver  or both the parties  apply the local dichotomic POVMs (in Eq. (\ref{POVM_i})) to activate the  teleportation fidelity with a non vanishing probability. We show that pre-processing sometimes allows us to enhance TF well beyond the classical limit (we call it as hidden TF) or to increase the TF beyond the initial fidelity which we refer it as enhanced TF. 
%. If the initial TF of the state is in classical region and after pre-processing it goes to quantum region, we call this increment as \textit{hidden} TF. On the other hand, if initially the TF is in the quantum region and after pre-processing it increases, we call this increment as \textit{enhanced} TF. 
Note that  if the post-processed state has TF below \(2/3\), we discard the state  and follow the best classical protocol.  As considered in the dense coding protocol,  we  define three quantities to monitor the action of pre-processing operations on TF. Specifically, we evaluate the optimal increase in TF, denoted by \(\mathbb{O}_{TF}\), average cost of optimum TF, \(\mathbb{A}_{TF}^1\) and maximal cost of average TF, \(\mathbb{A}_{TF}^2\) which are  respectively defined as in Eqs. (\ref{Optimumcost_DCC}), (\ref{averagecost1_DCC}) and (\ref{averagecost2_DCC})  by replacing \(\mathbb{C}\) by \(\mathbb{F}\).

We now consider the action of pre-processing when the shared state is a random pure two-qubit state.\\

\begin{figure}[!ht]
	\resizebox{8cm}{4.5cm}{\includegraphics{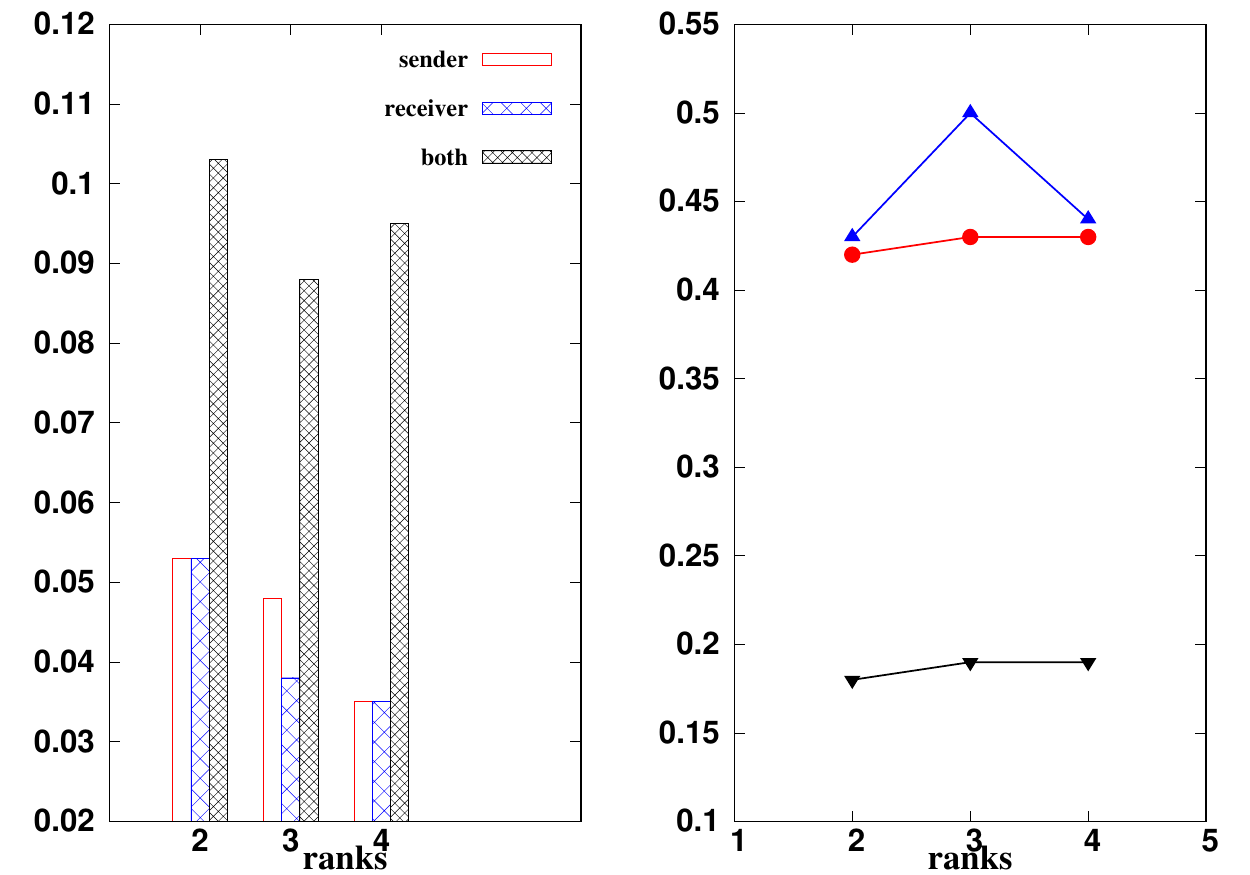}}
	\caption{\footnotesize (Color online)  (Left) The optimum increment in teleportation fidelity on average, denoted by \(\overline{\mathbb{O}}_{TF}\), vs. ranks of two-qubit states. (Right)  \(\overline{p}_{\mathbb{O}_{TF}}\) against ranks. Plots clearly show the trade off between the increment in TF and the success probability. All the axes are dimensionless. 
	%Red circular points correspond to the case when only sender is doing pre-processing, blue upper triangular points represents the case when only receiver is doing pre-processing and black lower triangular points denotes the case when both sender and receiver are doing pre-processing.
	}
\label{fig:maxprobTF}
\end{figure}

\begin{figure}[!ht]
	\resizebox{8cm}{5cm}{\includegraphics{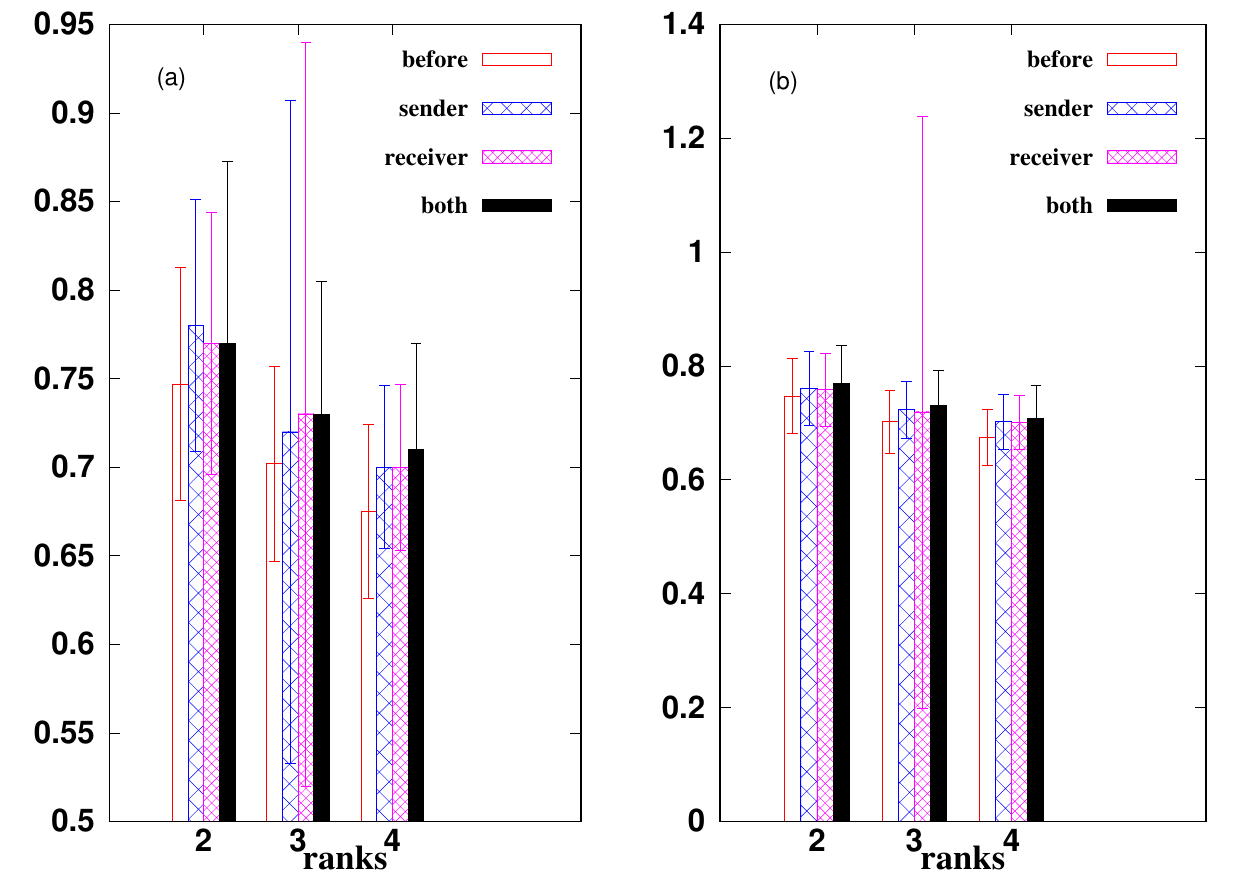}}
	\caption{ (Color online)   Left plot represents mean of average cost of optimum TF,	$\overline{\mathbb{A}}_{TF}^{1}$ while  right one is for $\overline{\mathbb{A}}_{TF}^{2}$ vs. ranks. SD are shown as error bars. All the axes are dimensionless. }
\label{fig:meanplotTF}
\end{figure}

\begin{table}[]
	\resizebox{0.5\textwidth}{!}{\begin{minipage}{0.6\textwidth}
			\caption{Mean of the average cost of optimum TF and maximal cost of average TF after applying POVMs. ``Before" again represents  TF  of random states on average for a given rank, without the action of local POVM. } 
			\label{meanSDTF}
			\centering
				\begin{tabular}{|l|l|l|l|l|l|l|l|}
				\hline
				& \multicolumn{1}{c|}{Before} & 
				\multicolumn{2}{c|}{Sender}                                                                                                                                                                                               & \multicolumn{2}{c|}{Receiver}                                                                                                                                                                                             & \multicolumn{2}{c|}{Both}                                                                                                                                                                                                 \\ \hline
				& \multicolumn{1}{c|}{}       & \multicolumn{1}{c|}{}                                                                                       & \multicolumn{1}{c|}{}                                                                                       & \multicolumn{1}{c|}{}                                                                                       & \multicolumn{1}{c|}{}                                                                                       & \multicolumn{1}{c|}{}                                                                                       & \multicolumn{1}{c|}{}                                                                                       \\ \hline
				& \multicolumn{1}{c|}{}       & \multicolumn{1}{c|}{$\overline{\mathbb{A}}_{TF}^{1}$} & \multicolumn{1}{c|}{$\overline{\mathbb{A}}_{TF}^{2}$} & \multicolumn{1}{c|}{$\overline{\mathbb{A}}_{TF}^{1}$} & \multicolumn{1}{c|}{$\overline{\mathbb{A}}_{TF}^{2}$} & \multicolumn{1}{c|}{$\overline{\mathbb{A}}_{TF}^{1}$} & \multicolumn{1}{c|}{$\overline{\mathbb{A}}_{TF}^{2}$} \\ \hline
				&                             &                                                                                                             &                                                                                                             &                                                                                                             &                                                                                                             &                                                                                                             &                                                                                                             \\ \hline
				Rank 2 & 0.747                       & 0.761                                                                                                       & 0.78                                                                                                        & 0.758                                                                                                       & 0.77                                                                                                        & 0.769                                                                                                       & 0.77                                                                                                        \\ \hline
				Rank 3 & 0.702                       & 0.723                                                                                                       & 0.72                                                                                                        & 0.719                                                                                                       & 0.73                                                                                                        & 0.73                                                                                                        & 0.73                                                                                                        \\ \hline
				Rank 4 & 0.675                       & 0.702                                                                                                       & 0.7                                                                                                         & 0.701                                                                                                       & 0.7                                                                                                         & 0.708                                                                                                       & 0.71                                                                                                        \\ \hline
			\end{tabular}
	\end{minipage}}
\end{table}

\begin{table}[]
	\resizebox{0.5\textwidth}{!}{\begin{minipage}{0.6\textwidth}
			\caption{Percentage of states showing non classical teleportation fidelity before and after the actions of POVM.   All other specifications are the same as in Table \ref{meanSDTF}. } 
			\label{percentageTF}
			\centering
		\begin{tabular}{|l|l|l|l|l|l|l|l|}
			\hline
			& \multicolumn{1}{c|}{Before} & \multicolumn{2}{c|}{Sender}                                                                                                                                                                                               & \multicolumn{2}{c|}{Receiver}                                                                                                                                                                                             & \multicolumn{2}{c|}{Both}                                                                                                                                                                                                 \\ \hline
			& \multicolumn{1}{c|}{}       & \multicolumn{1}{c|}{}                                                                                       & \multicolumn{1}{c|}{}                                                                                       & \multicolumn{1}{c|}{}                                                                                       & \multicolumn{1}{c|}{}                                                                                       & \multicolumn{1}{c|}{}                                                                                       & \multicolumn{1}{c|}{}                                                                                       \\ \hline
			& \multicolumn{1}{c|}{}       & \multicolumn{1}{c|}{$\overline{\mathbb{A}}_{TF}^{1}$} & \multicolumn{1}{c|}{$\overline{\mathbb{A}}_{TF}^{2}$} & \multicolumn{1}{c|}{$\overline{\mathbb{A}}_{TF}^{1}$} & \multicolumn{1}{c|}{$\overline{\mathbb{A}}_{TF}^{2}$} & \multicolumn{1}{c|}{$\overline{\mathbb{A}}_{TF}^{1}$} & \multicolumn{1}{c|}{$\overline{\mathbb{A}}_{TF}^{2}$} \\ \hline
			&                             &                                                                                                             &                                                                                                             &                                                                                                             &                                                                                                             &                                                                                                             &                                                                                                             \\ \hline
			Rank 2 & 89.96\%                     & 95.99\%                                                                                                     & 99.5\%                                                                                                      & 94.88\%                                                                                                     & 99.45\%                                                                                                     & 87.68\%                                                                                                     & 99.65\%                                                                                                     \\ \hline
			Rank 3 & 79.09\%                     & 87.51\%                                                                                                     & 99.26\%                                                                                                     & 79.3\%                                                                                                      & 99.17\%                                                                                                     & 86.03\%                                                                                                     & 99.58\%                                                                                                     \\ \hline
			Rank 4 & 52.04\%                     & 82.62\%                                                                                                     & 99\%                                                                                                        & 78.87\%                                                                                                     & 89.95\%                                                                                                     & 83.51\%                                                                                                     & 99.48\%                                                                                                     \\ \hline
		\end{tabular}
	\end{minipage}}
\end{table}

%\begin{figure}[!ht]
%	\resizebox{7cm}{5cm}{\includegraphics{new_mean_sd_tf.pdf}}
%	\caption{ (Color online)  The mean and standard deviation (shown as error bars) of teleportation fidelity of randomly generated states is plotted against the rank of the mixed state before and after pre-processing. The TF after pre-processing is quantified by average TF cost-2 in this case.}
%\label{fig:newmeanplotTF}
%\end{figure}

$\bullet$\emph { TF after POVM on arbitrary two-qubit density matrices. } The effectiveness of local pre-processing operations in enhancing the TF of two-qubit random states  is studied. In the two-qubit scenario, the optimal TF achievable from a shared two-qubit state is already known \cite{Verstraete, Badziag}. Here, we compare the optimal fidelity already known with the POVMs considered in this paper. 
%Thus comparing the efficacy of our pre-processing protocol. We first address the effect of increasing ranks and number of parties in revealing the hidden TF.

\begin{enumerate}
\item \textit{Efficacy of POVMs increases with  ranks. } For a fixed rank,  the number of states showing teleportation fidelity more than the classical bound increases probabilistically if the sender or the receiver or both perform local dichotomic POVMs (see Table \ref{percentageTF}). If one increases the rank, such increment is dramatic as quantified by \(\overline{\mathbb{O}}_{TF}\),  especially after the action of POVM by both the parties, as shown in  Fig. \ref{fig:maxprobTF}. Unlike DC protocol with two-qubits, the sender can also help to increase the TF by applying local POVM. 

\item We observe that the mean values of \(\mathbb{A}_{TF}^i, \, i=1,2\)  do not change so much  after the action of  POVMs which is  different than the one observed in case of the DC protocol (comparing Table \ref{meanSDTF} and Fig. \ref{fig:meanplotTF} with Table \ref{meanDCCtab} and  Fig. \ref{fig:meanplotDCC}).  

%\textcolor{red}{Please combine Fig. 8 and 9}. 

\item In general,  we observe that TF can be enhanced maximally when both the parties perform an optimal POVM although the probability of obtaining such outcome on average is less in this case compared to the one when the sender or the receiver performs POVM.  On contrary, we find that for rank-2 random states,  the average cost of optimum TF is more when  pre-processing is on the sender's side only than on both sides. This is possibly due to the fact that during averaging,  TF corresponding to some of the outcomes is very small and for both-sided POVMs, a number of such outcomes are more compared to the single-sided ones. 
\end{enumerate}

\section{Conclusion}
\label{Conclusion}

It is hard to emphasize the role of dense coding and teleportation  protocols to build a new arena of research dealing with quantum technologies. In laboratories,  perfect dense coding capacity (DCC) and teleportation fidelity cannot be  achieved due to the presence of different decohering factors and imperfections. Therefore, it is of utmost importance to devise a technique to restore quantum advantage as much as possible from low-performing states. It is usually done via pre-processing of channels which include distillation and  filtering processes. By using these techniques, specific protocols are known for a specific class of states or for two-qubits.  

In this work,  we characterized TF for random two-qubit states and DCC for random  two-, three- and four-qubits before any pre-processing.  We then showed that substantial  activation and enhancement in capacities, as well as fidelities, can happen after applying local pre-processing by the sender(s) and the receiver(s).   For rank-2 two-qubit and three-qubit states, we analytically found  that DCC of rank-2 states having the same amount of entanglement with pure two-qubit states and three-qubit generalized Greenberger-Horne-Zeilinger state is lower than that of the pure states, thereby establishing the fact that DCC and entanglement content of the shared states are not interconnected. We also proved that Werner state provides a lower bound on the DCC of rank-2 two-qubit states provided they possess the same amount of entanglement. Both upper and lower bounds obtained turned out to be true for any two- and three-qubit density matrices. Numerical simulations also showed that  the lower bound holds for rank-2 states after local pre-processing. 
% in the context of dense coding--Werner state provides the lower bound of DCC which also holds for preprocessed states revealing no direct connection between entanglement and DCC. For 3-qubit states upper bound is given by gGHZ states. 
We defined three distinct figures of merit to access the advantage of local  pre-processing. We  found  that the fraction of states exhibiting quantum advantage in DC and teleportation decreases with the increase of rank which can be overcome by means of local pre-processing operations before beginning  the protocols. In the case of teleportation, it is interesting to see that for rank-$3$ and rank-$4$ states, $93\%$ and $98\%$ of states showing non-classical TF does not violate Clauser-Horne-Shimony-Holt Bell inequality \cite{CHSH}.

%And states not useful in terms of mean fidelity are got activated after action of local POVM. 

\section*{acknowledgement}
We acknowledge the support from the Interdisciplinary Cyber Physical Systems (ICPS) program of the Department of Science and Technology (DST), India, Grant No.: DST/ICPS/QuST/Theme- 1/2019/23. We  acknowledge the use of \href{https://github.com/titaschanda/QIClib}{QIClib} -- a modern C++ library for general purpose quantum information processing and quantum computing (\url{https://titaschanda.github.io/QIClib}) and cluster computing facility at Harish-Chandra Research Institute.

\section*{Appendix: State dependent pre-processing by pure states}

In the case of pure two-qubit states, we know that all the states are dense codable as well as can give non classical teleportation fidelity. Any  pure two-qubit state  can be written in the Schmidt form \cite{nielsenchuang} as 
\begin{equation}
\rho^{SR} = \cos(\theta/2) |00\rangle + \sin(\theta/2)|11\rangle  
\label{r1schmidt}
\end{equation}
where $|0\rangle$ and $|1 \rangle $ are the orthonormal basis. Let us consider two situations, when the sender (the receiver) performs pre-processing and when both of them perform pre-processing. 
\begin{enumerate}
	\item \textit{Pre-processing by  sender  (receiver)  } Let us suppose  the sender $S$ ( the receiver $R$) performs the following pre-processing operations on its part of the qubit \cite{Gisin}. 
	\begin{equation}
P_{S}^+ = \begin{pmatrix}
\tan(\theta) && 0 \\
0 && 1
\end{pmatrix}
;P_{S}^- =  \begin{pmatrix}
\sqrt{1 - \tan^2(\theta)} && 0 \\
0 && 0.
\end{pmatrix}
\end{equation}
If $S$ gets the outcome '$+$',  the resultant state becomes a maximally entangled state whose dense coding capacity is $2$ and the TF is unity. Note that following the notations in Eq. (\ref{POVM_state1}), $\sqrt{E_{S}^+} = P_{S}^+$ and the normalized state after the action of pre-processing is given by Eq.(\ref{POVM_state1}). The success probability is $\text{tr}[(P_{S}^+ \otimes \openone_{R}).\rho^{SR}.(P_{S}^{+^{\dagger}} \otimes \openone_{R})]$. 
	
%	\item \textit{Only receiver does pre-processing:} In this case only the receiver $B_1$ performs the pre-processing operations on its part of the qubit. The pre-processing operations are same as in the case when only sender performs the pre-processing. The resultant state is a maximally entangled state when the outcome '+' occurs.  The normalized state after the action of pre-processing is given by Eq.(\ref{POVM_state1}). The success probability is $\text{tr}[(\openone_{A_1} \otimes P_{B_1}^+).\rho^{A_1B_1}.(\openone_{A_1} \otimes P_{B_1}^{+^{\dagger}})]$.
%	
	\item \textit{Both sender and receiver do pre-processing.} When both the sender and the receiver apply  the following operations on their part of the qubit:
	\begin{eqnarray}
P_{S}^+ &=& \begin{pmatrix}
\sqrt{\tan(\theta)} && 0 \\
0 && 1
\end{pmatrix}
;P_{S}^- =  \begin{pmatrix}
\sqrt{1 - \tan(\theta)} && 0 \\
0 && 0
\end{pmatrix} \nonumber \\
P_{R}^+ &=& \begin{pmatrix}
\sqrt{\tan(\theta)} && 0 \\
0 && 1
\end{pmatrix}
;P_{R}^- =  \begin{pmatrix}
\sqrt{1 - \tan(\theta)} && 0 \\
0 && 0.
\end{pmatrix}
\end{eqnarray}
and if they get the outcome '$++$', the output state is maximally entangled, giving maximal DCC and TF. As in Eq. (\ref{POVM_state1}), $\sqrt{E_{S}^+} = P_{S}^+$, $\sqrt{E_{R}^+} = P_{R}^+$ and the normalized state after the action of pre-processing is given by Eq.(\ref{POVM_state1}). The success probability is $\text{tr}[(P_{S}^+ \otimes P_{R}^+).\rho^{SR}.(P_{S}^{+^{\dagger}} \otimes P_{R}^{+^{\dagger}})]$. 
We find that the probability of success in both situations is equals to $2 (\sin \theta)^2$. We find that although the above pre-processing  leads to  a higher \(\mathbb{O}_{DCC}\) and \(\mathbb{O}_{TF}\), compared to the state-independent method described in the paper, the average cost of optimum dense coding capacity, as well as the maximal cost of average DC (TF), turn out to be higher in the state-independent POVMs (see Figs. \ref{fig:meanplotDCC} and \ref{fig:newmeanplotDCC}). 
	
\end{enumerate}

\bibliography{bib1}
\bibliographystyle{apsrev4-1}

%\end{thebibliography}

\end{document}